\documentclass[11pt]{article}

\RequirePackage[OT1]{fontenc}
\RequirePackage{amsthm,amsmath,amssymb}
\RequirePackage{natbib}
\RequirePackage[colorlinks,citecolor=blue,urlcolor=blue]{hyperref}

\usepackage{amssymb, amsthm}
\usepackage[normalem]{ulem}
\usepackage[hang,flushmargin]{footmisc}
\usepackage{enumitem}  
\usepackage[hang,flushmargin]{footmisc}
\usepackage{lipsum}
\usepackage{changepage}
\usepackage{colortbl}
\usepackage[dvipsnames]{xcolor}
\usepackage{bbm,mathtools} 
\usepackage{graphicx,adjustbox}
\usepackage{multirow,makecell,subcaption}
\usepackage[linesnumbered,ruled,vlined]{algorithm2e}
\usepackage{tabularx,algorithmic}
\usepackage{color,float,xcolor,colortbl,multirow}
\usepackage{epstopdf}

\definecolor{lightgray}{gray}{0.85}
\definecolor{middlegray}{gray}{0.75}
\definecolor{darkgray}{gray}{0.65}

\makeatletter
\newcommand{\algorithmfootnote}[2][\footnotesize]{%
  \let\old@algocf@finish\@algocf@finish
  \def\@algocf@finish{\old@algocf@finish
    \leavevmode\rlap{\begin{minipage}{\linewidth}
    #1#2
    \end{minipage}}%
  }%
}
\makeatother

\SetKwRepeat{Do}{do}{while}%
\SetKwInOut{Output}{output}
\SetKwInput{Initialization}{Initialization}
\SetKwProg{Fn}{Function}{:}{}
\def\bm#1{\mbox{\boldmath $#1$}}
\def\thick#1{\hbox{\rlap{$#1$}\kern0.25pt\rlap{$#1$}\kern0.25pt$#1$}}

\newtheorem{theorem}{Theorem}
\newtheorem{lemma}{Lemma}

\newtheorem{proposition}{Proposition}
\theoremstyle{definition}
\newtheorem{definition}{Definition}

\newtheorem{remark}{Remark}
\def\bm#1{\mbox{\boldmath $#1$}}

\def\E{\mathbb E}

\def\Cov{{\rm Cov}}
\def\BSIGMA{{\bm\Sigma}}

\newcommand{\bS}{\bm{S}}

\newcommand{\bX}{\bm{X}}

\newcommand{\bg}{\bm{g}}

\newcommand{\bC}{\bm{C}}
\newcommand{\bR}{\bm{R}}

\newcommand{\mS}{\mathcal{S}}

\newcommand{\bh}{h}

\newcommand{\mM}{\mathcal{M}}
\newcommand{\mMS}{\mathcal{M}\cup \mathcal{S}}
\newcommand{\bD}{\bm{D}}
\newcommand{\bA}{\bm{A}}

\newcommand{\indep}{\perp \!\!\! \perp}

\newcommand{\bSigma}{\bm{\Sigma}}

\newcommand{\bOmega}{\bm{\Omega}}

\newcommand{\bW}{\bm{W}}

\newcommand{\op}{\mathrm{op}}
\def\mA{\mathcal{A}}

\def\Var{\mbox{Var}}

\def\Cov{\mbox{Cov}}

\newdimen\biblioindent \biblioindent=30pt

 at 7truept

\def\bH{\bm{H}}

\DeclareMathOperator*{\argmin}{argmin}

\def\beq{\begin{equation}}
	\def\eeq{\end{equation}}
\def\beqn{\begin{eqnarray}}
	\def\eeqn{\end{eqnarray}}
\def\beqnn{\begin{eqnarray*}}
	\def\eeqnn{\end{eqnarray*}}

\theoremstyle{remark}  

\begin{document}
	
	
	\renewcommand{\baselinestretch}{1.5}
	
	
	
	$\ $\par
	
 \thispagestyle{plain}
 
	\fontsize{12}{14pt plus.8pt minus .6pt}\selectfont \vspace{0.8pc}
	\centerline{\large\bf Hypothesis Testing for Penalized Estimating Equations with 
}
	\vspace{2pt} 
	\centerline{\large\bf Cross-Fitted Covariance Calibration}
	\vspace{.4cm} 
	\centerline{Jing Zhou$^{1}$, Zhe Zhang$^{2}$}
	\vspace{2pt}
	\centerline{\it $^{1}$Department of Mathematics, University of Manchester}
	\centerline{\it $^{2}$Perelman School of Medicine, University of Pennsylvania}
	\centerline{\it Email: jing.zhou@manchester.ac.uk, zhe.zhang@pennmedicine.upenn.edu}
	\vspace{.4cm} 
 
	\vspace{.55cm} \fontsize{12}{11.5pt plus.8pt minus.6pt}\selectfont
	

	\begin{quotation}
     \noindent {\it Abstract:} We study hypothesis testing for penalized estimators in settings where the full marginal distribution of a multivariate response is difficult to specify, such as longitudinal data with correlated measurements or high-dimensional heteroscedastic regression. Assuming that the conditional mean model is correctly specified, we establish that the penalized estimating equations admit a $\sqrt{n}$-consistent solution, even when the working covariance structure is misspecified. Our inferential target is a low-dimensional subvector of parameters associated with the mean model. We show that the resulting test statistic converges to a $\chi^2$ distribution, and that its asymptotic power depends on the nuisance covariance function. To mitigate this dependence, we propose estimating the covariance function via cross-fitting, which provides a calibrated and robust procedure for inference.
  \vspace{9pt}
		\par
\end{quotation}\par
	
	\def\thefigure{\arabic{figure}}
	\def\thetable{\arabic{table}}
	
	\renewcommand{\theequation}{\thesection.\arabic{equation}}

	\fontsize{12}{14pt plus.8pt minus .6pt}\selectfont


 \section{Introduction}

In high-dimensional models, estimation and inference for parameters associated with the conditional mean are complicated by heteroscedasticity and unknown covariance structures. In such settings, specifying the joint density of a multivariate response vector is often challenging, and the properties of likelihood-based estimators that ignore heteroscedasticity or covariate-dependent covariance structures are underexplored. Moreover, misspecification of the covariance structure can lead to efficiency loss \citep{aitkin1935least} and invalid inference \citep{White1982,mackinnon1985some}. 

Consider a sample $\{(Y_i, \bm X_i^\top)\}_{i=1}^n$, where $\bm X_i \in \mathbbm R^{p\times l}$ and $Y_i \in \mathbbm R^l$. Rather than fully specifying the data-generating process, we consider estimation based on the following conditional mean structure
\begin{equation}\label{eq:conditional_mean}
    \mathbb{E}(Y_i \mid \bX_i) 
    = \bg(\bX^\top_i \beta_0)
    = \bigl(g(X_{i1}^\top \beta_0), \ldots, g(X_{il}^\top \beta_0)\bigr)^\top,
\end{equation}
where $g(\cdot)$ is a known link function and $\beta_0 \in \mathbb{R}^p$ is an unknown $s$-sparse parameter vector of primary interest. Our goal is to conduct estimation and inference for $\beta_0$ in the high-dimensional regime $p > n$, based on the mean specification \eqref{eq:conditional_mean}, without imposing restrictive assumptions on the conditional covariance of $Y_i$. Such specification allows the observations $i \in [n]$ to be independent but not identically distributed.

When the response $Y_i$'s are scalar-valued, a classical approach for 
accommodating dispersion and heteroscedasticity is through quasi-likelihood methods \citep{wedderburn1974quasi,McCullaghNelder1989}. Such estimation requires solving quasi-score equations, which specify the conditional mean and variance functions. Typically, the variance function depends on known functions of the mean up to a scalar dispersion parameter. When the quasi-score function is integrable, a scalar quasi-likelihood whose maximizer admits a well-developed asymptotic theory.


However, extending quasi-likelihood methods to multivariate responses is substantially more challenging. In general, the quasi-score for a multivariate mean does not coincide with the gradient of any scalar objective function unless restrictive integrability conditions are imposed on the covariance structure \citep{godambe1960optimum,liang1986longitudinal}. As a result, likelihood- or quasi-likelihood-based formulations are typically unavailable (or become cumbersome) when the response dimension is greater than one, and the covariance is unknown or covariate-dependent. In contrast, when $l=1$ (a univariate response), quasi-likelihood reduces to a standard scalar criterion and can be combined naturally with regularization; see, for example, the penalized quasi-likelihood developments in \citet{mammen1997penalized} and related high-dimensional treatments such as \citet{shi2019linear}. This motivates our use of estimating-equation methods, which avoid the need to specify a fully integrable multivariate quasi-likelihood while still allowing flexible modeling of $\bSigma(\cdot)$.

To address this limitation, we adopt an estimating-equations-based framework that does not require the existence of a (quasi-)likelihood function. Specifically, generalized estimating equations (GEE) can be constructed by specifying the conditional mean and a working covariance structure. This approach ensures consistency and asymptotic normality, even in the presence of covariance misspecification \citep{liang1986longitudinal,godambe1985foundations}. 

Various approaches to estimating $\beta_0$ within the framework of generalized estimating equations (GEE) have been developed for longitudinal and clustered data. These methods typically decompose the conditional covariance structure into a marginal variance component—often motivated by quasi-likelihood theory for exponential family models—and a low-dimensional working correlation matrix. A common strategy specifies the working correlation as a linear combination of basis matrices; see, for example, \citep{qu2000improving}. Such working correlations are typically assumed to be deterministic, common across subjects, and independent of covariates.

However, such specifications may be restrictive in practice. In particular, the marginal variance is typically modeled as a known function of the mean and may fail to accommodate more general forms of marginal heteroscedasticity. Moreover, standard working correlation specifications do not explicitly accommodate dependence structures that vary with covariates. However, there is evidence from empirical studies that correlation and marginal variance may both be covariate-dependent; see, for example, the labour income study in \citet{meghir2004income,guvenen2009empirical}; evolutionary and behavioural ecological studies in \citet{nakagawa2025location}.

To encourage sparsity in estimating $\beta_0$ within the GEE framework, penalization methods have been proposed; see, for example, \citet{wang2012penalized,fang2020test}. These works focus on longitudinal data and model the covariance structure using the classical decomposition into a marginal variance function motivated by an exponential family and a deterministic working correlation. Throughout this paper, we do not assume a longitudinal data structure; instead, the GEE framework is used solely as a convenient vehicle for constructing estimating equations under general forms of covariance misspecification.

Similar to the partially penalized GEE considered in \citet{fang2020test}, let $\mM \subseteq \{1,\ldots,p\}$ be a subset of indices with cardinality $|\mM| = m$, and denote the corresponding subvector of $\beta_0$ by $\beta_{0,\mM}$. The testing problem of interest in this paper is
\begin{equation}\label{eq: hypothesis}
    H_0: \bm C \beta_{0,\mM} = t 
    \quad \text{versus} \quad 
    H_a: \bm C \beta_{0,\mM} \neq t,
\end{equation}
where $\bm C$ is a given $r \times m$ matrix with full row rank and $t$ is an $r$-dimensional vector, with $r \le m$. A special case of interest is $\bm C = \bm I_m$, the $m \times m$ identity matrix, and $t = \bm 0_r$, the $r$-dimensional zero vector.

In contrast to \citet{fang2020test}, our objective is not to model longitudinal data, but rather to understand the behavior of estimators of $\beta_0$ when the underlying data-generating process is more complex. Specifically, we allow the conditional covariance structure to depend on covariates through an unknown, potentially nonlinear function. That is,
\begin{equation}\label{eq: conditional covariance matrix}
    \Cov (Y_i \mid \bX_i)
    =
    \bSigma(\bX_{i,\mA})
    \in \mathbb{R}^{l \times l},
\end{equation}
where $\mA \subseteq \{1,\ldots,p\}$ denotes an active subset of covariates on which $\bSigma(\cdot)$ depends, $\bX_{i, \mA}\coloneqq (X_{i1,\mA},\ldots,X_{il,\mA})\in\mathbb{R}^{l\times|\mA|}$ is formed by rows of $\bX_i$ in $\mA$, and $\operatorname{vec}(\bX_{i,\mA}) \in \mathbb{R}^{l\cdot|\mA|}$ denotes its vectorized version satisfying Assumption~\ref{cond: active set covariance matrix}. For notational simplicity, we write $\bSigma(\bX_{i,\mA})$ instead of $\bSigma(\operatorname{vec}(\bX_{i,\mA}))$. Estimation in this setting is nontrivial, as \eqref{eq: conditional covariance matrix} allows the responses $Y_i$'s to be independent but not identically distributed. As is relevant to the labor income research in \citet{guvenen2009empirical, meghir2004income}, modeling the variance function is of independent interest with practical interpretation for empirical studies.


We propose a plug-in estimator of the working covariance function $\bSigma(\cdot)$, motivated by inverse-variance weighting to gain efficiency under heteroscedasticity. In linear models, this leads to GLS/WLS and feasible GLS, in which the error variance is estimated from an initial fit (e.g., \citealp{amemiya1973regression}); see also \citet{andrews1986note} for further discussion. Related sandwich variance estimators yield inference that remains valid under unknown heteroscedasticity \citep{white1980heteroskedasticity}. Our construction is also connected to variance-function estimation and joint mean--dispersion modeling in fixed-$p$ settings \citep{davidian1987variance}, and to recent location--scale methods that adapt to conditional heteroscedasticity \citep{spady2018simultaneous,young2024sandwich}. 
Overall, $\hat{\bSigma}(\cdot)$ provides a practical way to incorporate these weighting ideas into our estimating-equation framework while keeping the focus on the statistical properties of the resulting estimator.

\section{Model setup}


\subsection{Estimating equations}

We denote a sample of $n$ observations as $\{(Y_i, \bm X_i)\}_{i=1}^n$ where $\bm X_i$ are independent and identicially distributed. Further, due to the specifications in \eqref{eq:conditional_mean} and \eqref{eq: conditional covariance matrix}, the copies $Y_i$'s are independent, but an identical distribution is not necessary.

If $\bSigma(\cdot)$ and the active set $\mA$ were known, the estimating equations can be defined as 
\begin{eqnarray}\label{eq: estimating equations general}
    U_n(\beta) = \frac{1}{n} \sum_{i=1}^n  U_i(\beta) \in \mathbb{R}^p,
\end{eqnarray}
where $U_i(\beta) = \bX_i \bD_i(\beta) \bSigma(\bX_{i,\mA})^{-1} \left\{Y_i - \bg\left(\bX_i^\top\beta\right)\right\}$. Let $\dot g (\cdot)$ denote the first derivative of $g(\cdot)$, $\bD_i(\beta) = \operatorname{diag}(\dot{g}(X_{i1}^\top \beta), \ldots, \dot{g}(X_{il}^\top \beta))$. We refer to $U_n(\beta)$ as the \emph{oracle estimating equations}, since they are constructed using the true covariance matrix 
$\bSigma(\cdot)$, which is unknown in practice.

The partially penalized estimating equations $U^{\rm p}_n(\beta)$ are defined as
\begin{equation}\label{eq: partial penalized estimating equation}
U^{\rm p}_n(\beta) \coloneqq U_n (\beta) + \partial \rho_\lambda (\beta;\mM),
\end{equation}
where $\partial \rho_\lambda (\beta; \mM) = (\dot\rho_\lambda(\beta_1; \mM), \ldots, \dot\rho_\lambda(\beta_p; \mM))^\top$, $\dot\rho_\lambda(\beta_1; \mM) \coloneqq \dot\rho_\lambda(\beta_j) I\{ j \notin \mM\}$, and $\dot \rho_\lambda(\beta_j)$ satisfies Assumption~\ref{cond: nonconvex regularizer}. The two most common nonconvex penalties are the SCAD \citep{fan2001variable} with $\dot\rho_{{\rm scad},\lambda}(t; a) = \lambda [I(|t| \le \lambda) + \frac{(a\lambda - |t|)_+ }{(a -1)\lambda} I(|t| > \lambda)]$ and MCP \citep{zhang2010mcp} with $\dot\rho_{{\rm mcp},\lambda}(t; a) = (\lambda - |t|/a)_+$, where $a > 0$ is some fixed constant controlling the shrinkage of large coefficients, see Assumption~\ref{cond: large value no thresholding}.

By definition, \eqref{eq: partial penalized estimating equation} implies that no penalty is imposed on the subvector $\beta_{0,\mM}$, which is of 
primary interest for hypothesis testing. An estimator of the $p$-dimensional sparse parameter vector $\beta_0$, denoted by $\tilde\beta$, is obtained by solving
\begin{equation}\label{eq: partially penalized estimating equations}
    \bm{0} \in U^{\rm p}_n(\beta).
\end{equation}
To compute solutions to the regularized estimating equations in our numerical studies, we use a self-implemented algorithm based on the computational framework of the \texttt{PGEE} package in \textsf{R}. Because this paper focuses on theory and inference for our estimator, we omit further algorithmic details.

The oracle estimating equations $U_n(\beta)$ involve an unknown covariance matrix $\bSigma(\cdot)$, which complicates constructing the estimator. In practice, people specify a working covariance structure $\check\bSigma(\cdot)$, and we discuss the conditions on $\bSigma(\cdot)$ and $\check\bSigma(\cdot)$ in order to guarantee a $\sqrt{n}$-consistent estimator in Proposition~\ref{prop: consistency}.

\begin{proposition}\label{prop: consistency}
   Under Assumptions~\ref{cond: fixed and bounded design}, \ref{cond: conditional mean function}, \ref{cond: residuals moments and subgaussian}, \ref{cond: eigenvalue max min finite}, \ref{cond: variance function bounded}, and \ref{cond: log p and n order}, for $
\lambda_n\asymp\max\left\{\frac{\log n}{\sqrt{n}}\sqrt{\log p}, \frac{p (m+s)^{5/2} (\log n)^3}{n^{3/2}}\right\}$, there exists a partially penalized solution $\tilde{\beta} = \left(\tilde{\beta}^\top_{\mMS},\tilde{\beta}^\top_{\left(\mMS\right)^c}\right)^\top$ to \eqref{eq: partially penalized estimating equations}, which satisfies the following consistency properties
        \begin{equation}\label{eq: equations negligible 1}
            P\left\{|U^{\rm p}_{n, j} (\tilde\beta)| = \bm 0, j \in \mMS \right\} \to 1
        \end{equation}
        \begin{equation}\label{eq: equations negligible 2}
            P \left\{ |U_{n, j}^{\rm p} (\tilde{\beta})| \le \frac{\lambda_n}{\log n} , j \in (\mMS)^c \right\} \to 1,
        \end{equation}
        \begin{equation}\label{eq: prop section consistency}
        P\left(\tilde{\beta}_{\left(\mMS\right)^c} = \bm 0\right) \rightarrow 1,
        \end{equation}
        \begin{equation}\label{eq: prop active set estimation consistency}
        \|\tilde{\beta}_{\mMS} - \beta_{0,\mMS}\|_2 = O_P\left(\sqrt{(s+m)/{n}}\right).
        \end{equation}
\end{proposition}

\begin{remark}


Proposition~\ref{prop: consistency} establishes consistency of the solution $\tilde{\beta}$ to the penalized estimating equations in \eqref{eq: partially penalized estimating equations}; the proof is given in Section~\ref{ssec: proof of proposition consistency}. The role of the (unknown) covariance function $\bSigma(\cdot)$ is reflected through a mild regularity condition, stated in Assumption~\ref{cond: variance function bounded}, which prevents degeneracy of the weighting scheme. In particular, we require the diagonal weights in $\bSigma(\cdot)^{-1}$ to be uniformly bounded above, i.e.,
\[
\sup_{x}\max_{k}\big[\bSigma(x)^{-1}\big]_{kk}<\infty.
\]
This condition is analogous to requiring the conditional variance to be bounded away from zero in heteroscedastic regression; see, e.g., \citet{spady2018simultaneous}.

Importantly, the proof of Proposition~\ref{prop: consistency} does not rely on the correct specification of the working covariance. Consistency for $\beta_0$ continues to hold (under a correctly specified mean model) provided the inverse working covariance $\check{\bSigma}(\cdot)^{-1}$ satisfies the same uniform boundedness condition,
\[
\sup_{x}\max_{k}\big[\check{\bSigma}(x)^{-1}\big]_{kk}<\infty,
\]
even if $\check{\bSigma}(\cdot)\neq \bSigma(\cdot)$. We use this robustness property to justify the initial estimator under a misspecified working covariance and, subsequently, the cross-fitted estimator based on an estimated covariance function.

\end{remark}

\section{cross-fitted estimation with estimated covariance function}
 By Proposition~\ref{prop: consistency}, we can conclude that misspecification of the working covariance function $\check{\bSigma}(\cdot)$ does not affect the consistency of the estimator $\check{\beta}$. However, our primary interest lies in the testing problem in~\eqref{eq: hypothesis}, which depends on the asymptotic distribution of estimators of $\beta_0$. Since the true covariance function $\bSigma(\cdot)$ is unknown, it is natural to employ a data-driven estimator for this nuisance component. Care must be taken to avoid dependence between the estimated covariance function $\hat{\bSigma}(\cdot)$ and the estimating equations, as such dependence can introduce additional first-order terms in the asymptotic expansion and generally prevent the estimator from achieving $\sqrt{n}$-asymptotic normality without further restrictive conditions. To eliminate this bias and restore first-order orthogonality between the estimating equations and nuisance estimation error, we adopt a cross-fitting strategy \citep{chernozhukov2018DML}.

The cross-fitting workflow 
proceeds as follows. We split the dataset $\{(Y_i, \bX_i)\}_{i=1}^n$ into two disjoint subsamples. The two subsamples have index sets $\mathcal I_1$ and $\mathcal{I}_2$, such that $\mathcal I_1 \cup \mathcal I_2  = [n]$, $\mathcal I_1 \cap \mathcal I_2 = \varnothing$, $|\mathcal I_1| = \lfloor n/2 \rfloor$, and $|\mathcal I_2| = n - |\mathcal I_1|$. Each subsample is used to obtain an initial estimator, defined in \eqref{eq: population residuals definition} and denoted by $\check{\beta}^{(1)}$ and $\check{\beta}^{(2)}$, respectively, using the working covariance specification. By Proposition~\ref{prop: consistency}, these initial estimators are consistent. In Theorem~\ref{thm: covariance matrix}, we show that the residuals based on the first-stage estimators, denoted by $R_i (\check\beta^{(1)})$ and $R_i (\check\beta^{(2)})$, are sufficiently accurate proxies for $R_i (\beta_0)$ defined in \eqref{eq: population residuals definition}. Further, we compute residuals and use them to estimate the covariance functions $\hat{\bSigma}^{(1)}$ and $\hat{\bSigma}^{(2)}$ via Algorithm~\ref{algorithm: variance function update}. The estimated covariance functions are then cross-fitted and used as plug-in inputs in the updated estimating equations to obtain the second-stage estimators $\hat{\beta}^{(2)}$ and $\hat{\beta}^{(1)}$, respectively. The cross-fitted estimator is then aggregated and defined as
\begin{eqnarray}\label{eq: two stage estimator}
    \hat\beta = \frac{\hat\beta^{(1)} + \hat\beta^{(2)}}{2}.
\end{eqnarray}
The cross-fitted estimator $\hat\beta$ will be used to construct the test statistic in Section~\ref{sec: asymptotic normality}.

 \subsection{Covariance function}\label{ssec: variance function estimation}

To estimate the unknown function in \eqref{eq: conditional covariance matrix}, we define the residual associated with the unknown parameter $\beta$ as follows
\begin{equation}\label{eq: population residuals definition}
    R_{i} (\beta) \coloneqq Y_i - \bg (\bX_i^\top \beta),
\end{equation}
where the $R_i (\beta) = (R_{i1} (\beta), \ldots, R_{il}(\beta))^\top  = (Y_{i1} - g(X_{i1}^\top\beta), \ldots, Y_{il} - g(X_{il}^\top \beta))^\top$.
By \eqref{eq:conditional_mean},
$R_i(\beta_0)=Y_i-\bg(\bX_i^\top\beta_0)=Y_i-\E(Y_i\mid\bX_i)\in\mathbb{R}^l$,
which corresponds to the centered response. Further, we impose $\E[R_i(\beta_0)] = \bm 0$ in Assumption~\ref{cond: residuals moments and subgaussian}, which ensures that the conditional covariance function $\bSigma(\cdot)$ is identifiable from second moments.


Given the active set $\mA$ and the parameter $\beta_0$, the covariance function $\bSigma(x;\beta_0) : \mathbbm R^{l \cdot|\mA|} \mapsto \mathbbm R^{l \times l}$ may be estimated using a variety of existing nonparametric methods that model the nonlinear association between $R_i(\beta_0)R_i(\beta_0)^\top$ and the covariates $\bX_{i,\mA}$.
The choice of a specific covariance estimator is not central to our contribution, and alternative constructions could be employed.

For concreteness, we adopt a multivariate extension of the covariance function estimator proposed by \citet{yin2010nonparametric}, which considers the case $x\in\mathbb{R}$. Here $x \in \mathbb{R}^{l \cdot|\mA|}$ denotes the vectorized covariate argument corresponding to $\operatorname{vec}(\bX_{i,\mA})$. A technical requirement for our subsequent analysis is pointwise invertibility of the estimated covariance matrix. This is ensured by using a common bandwidth parameter across all components of the estimator. The resulting nonparametric estimator, for $q = 1, 2$, is defined as
\begin{equation}\label{eq: nonparametric covariance function}
    \hat\bSigma^{(q)}(x;\beta_0)
    =
    \sum_{i \in \mathcal{I}_q}
    \frac{K_H\left(\operatorname{vec}(\bX_{i,\mA})-x\right)}
         {\sum_{i' \in \mathcal{I}_q} K_H\left(\operatorname{vec}(\bX_{i,\mA})-x\right)}
    R_i(\beta_0)R_i(\beta_0)^\top,
\end{equation}
where the kernel function $K_H(u)=|H|^{-1/2}K(H^{-1/2}u)$ and the symmetric positive definite bandwidth matrix $H =  h^2 I_{l \cdot |\mathcal{A}|}$. Following the standard multivariate kernel regression result \citep{gine2002rates}, by choosing the bandwidth $h \asymp n^{-1/(4\nu + 2l \cdot |\mathcal{A}|)}$, it holds that
\begin{eqnarray}\label{eq: kernel estimator rate}
    \sup_{x} \|\hat\bSigma^{(q)} (x; \beta_0) - \bSigma (x; \beta_0)\|_F = O_P\left(\left(\frac{\log n}{n}\right)^{\frac{\nu}{4\nu + 2l \cdot |\mathcal{A}|}}\right).
\end{eqnarray}

\subsection{Plug-in estimators for $\hat\bSigma^{(q)}(x;\beta_0)$ }
The estimator in \eqref{eq: nonparametric covariance function} relies on knowledge of the active set $\mathcal{A}$ and the residuals $R_i(\check\beta)$ to model the unknown covariance function $\bSigma(\cdot)$. We describe the construction of a plug-in estimator $\hat{\mathcal{A}}$ in Section~\ref{sssec: plug in of covariance active set}. 

Section~\ref{sssec: covariance function full estimation} establishes the consistency of the residuals $R_i(\beta_0)$ and the estimated active set $\hat{\mathcal{A}}$ in Theorem~\ref{thm: covariance matrix}, with the proof deferred to Section~\ref{ssec: proof of consistency of covariance matrix}. In addition, Algorithm~\ref{algorithm: variance function update} provides pseudocode for estimating the covariance function.

Algorithm~\ref{algorithm: variance function update} is implemented separately on the two subsamples indexed by $\mathcal{I}_q$, $q = 1,2$. For notational simplicity, the subsample index $q$ is suppressed in Sections~\ref{sssec: plug in of covariance active set} and~\ref{sssec: covariance function full estimation}. All results for the full sample continue to hold for each subsample by replacing the sample size $n$ with $|\mathcal{I}_q|$ and by applying Algorithm~\ref{algorithm: variance function update} independently to $\{(Y_i, \bX_i)\}_{i \in \mathcal{I}_q}$.

\subsubsection{Selection of the active set $\mA$}\label{sssec: plug in of covariance active set}

We demonstrate the idea of selecting the active set $\mA_k$ for a single measurement
$k \in [l]$. Since the active set is common across measurements, it holds that
$\mA = \mA_k$ for all $k \in [l]$.
\\
Let $(Y_k, X_k)$ denote a generic copy of $(Y_{ik}, X_{ik})$, where
$X_k \in \mathbbm R^p$.
Similar to \eqref{eq: population residuals definition}, we define the population level residual at $k$ as
\[
R_k(\beta_0) = Y_k - g(X_k^{\top}\beta_0)
\]
and its conditional distribution given $X_k$, denoted by
$F_{R_k \mid X_k}$.
Our goal is to identify the active set $\mA_k$ of covariates that influence the
conditional distribution of $R_k(\beta_0)$ given $X_k$.
Specifically, we define $\mA_k$ such that the following conditional independence holds, that is,
\[
R_k(\beta_0) \indep X_{k,\mA_k^c} \mid X_{k,\mA_k},
\]
where $\mA_k^c$ denotes the complement of $\mA_k$, consisting of irrelevant covariates.
\\
To characterize this dependence structure, we consider the central subspace
$\mathcal{S}_{R_k \mid X_k}$, defined as the minimal subspace of $\mathbbm R^p$
that captures the dependence between the residual $R_k(\beta_0)$ and the covariates
$X_k$.
Formally, $\mathcal{S}_{R_k \mid X_k}$ is the column space of a $p \times d$ matrix
$\bm B$, with $d \le p$, such that
\[
R_k(\beta_0) \indep X_k \mid \bm B^{\top} X_k.
\]
Under Assumption~\ref{cond: fixed and bounded design}, this central subspace is uniquely defined, need references


To facilitate the identification of $\mathcal{S}_{R_k|X_k}$, we rely on the Least Squares Condition (LSC). This condition states that for any function $f(r): \mathbbm R \to \mathbbm R$, the resulting vector lies within the central subspace
\begin{equation}\label{eq: least squares condition}
{\bSigma}_{X_k}^{-1}\mbox{Cov}(X_k,f(R))\in \mathcal{S}_{R_k|X_k},  
\end{equation}
where $\bSigma_{X_k} = \E [X_{k} X_{k}^\top]$ the population covariance matrix of $X_{k}$. 

The condition \eqref{eq: least squares condition} motivates a constructive approach by using a series of basis functions $f_v$ indexed by $v \in [h]$ to approximate the unknown correlation between the covariates $X_{kj}$'s and the residual $R_k(\beta_0)$. The correlations of the covariates on the $v$th basis are defined as the solution to an ordinary least squares problem as follows 
\begin{eqnarray}\label{eq: theta population definition}
\theta^v_{k} \coloneqq \arg\min_{\theta \in \mathbbm R^p}\E[\{f_v(R_k(\beta_0)^2)-X_k^{\top}\theta\}^2],
\end{eqnarray}
where $\theta^v_{k} = (\theta^v_{k1}, \ldots, \theta^v_{kp})^\top$, and the $j$th element of $\theta^v_{k}$ is denoted by $\theta^v_{kj}$ and measures the functional dependence of covariate $X_{kj}$ and the $v$th basis of the residual $f_v(R_k(\beta_0))$. The idea is that, if $j\in\mA^c$, $X_k^{\top}\theta^v_{k}, v \in [h]$ must not involve $X_{kj}$. Thus, $R_k(\beta_0)\indep X_k| X_{\mA_k}$ and $R_k(\beta_0)\indep X_k|\bm{B}_k^{\top} X_k$, where $\bm{B}_k = (\theta^1_{k},\cdots, \theta_{k}^h) = \begin{pmatrix}
\theta^1_{k1} & \cdots & \theta^v_{k1} & \cdots & \theta^h_{k1} \\
\theta^1_{k2} & \cdots & \theta^v_{k2} & \cdots & \theta^h_{k2} \\
\vdots & \vdots &  \vdots & \vdots& \vdots \\
\theta^1_{kj} & \cdots & \theta^v_{kj} & \cdots & \theta^h_{kj} \\
\vdots & \vdots &  \vdots & \vdots& \vdots \\
\theta^1_{kp} & \cdots & \theta^v_{kp} & \cdots & \theta^h_{kp} \\
\end{pmatrix}
\in \mathbbm R^{p\times h}$. We denote the $j$th row of $\bm B_k$ by $B_{kj} = (\theta_{kj}^1, \ldots, \theta_{kj}^h)$ and the vector excluding the $j$th element of $\theta_{k}^v$ by $\theta^v_{k,-j} \coloneqq (\theta_{k1}^v, \ldots, \theta_{k, j-1}^v, \theta_{k, j+1}^v, \ldots, \theta_{kp}^v) \in \mathbbm R^{p-1}$.

We make the crucial coverage condition assumption, which is widely adopted \citep{li1991sliced,guo2025model,li2018sufficient} that the constructively found subspace is equivalent to the central subspace, that is,
$$\mbox{Span}(\bm{B}_k)=\mathcal{S}_{R_k|X_k}.$$
In addition, it holds that $\sum_{v=1}^h |\theta_{kj}^v|>0$ for $j\in\mA_k$ and $\sum_{v=1}^h |\theta_{kj}^v|=0$ for $j\in\mA_k^c$. 
Since $X_k$ is a high-dimensional random vector, motivated by the decorrelated score method to take the spurious correlation between the covariates into account, we consider the following ``correlation coefficient" 
$$\gamma_{kj} \coloneqq \argmin_{\gamma \in \mathbbm R^{p-1}} \E\left( X_{kj} - X_{k,-j}^{\top}\gamma\right)^2,$$
where $X_{k,-j} \in \mathbbm R^{p-1}$ is the covariate vector excluding the $j$th covariate $X_{kj}$.

With the above specification, we define the test statistic as follows
\begin{eqnarray}\label{eq: test statistic covariance active set selection}
    W_{kj} = \bar{S}_{kj}^\top \hat\Omega_{kj}^{-1} \bar{S}_{kj},
\end{eqnarray}
where $\bar{S}_{kj} = \frac{1}{\sqrt{n}} \sum_{i =1}^n S_{ikj}$, $S_{ikj} = (S_{ikj}^1, \ldots, S_{ikj}^h)^\top$, $S_{ikj}^v = (X_{ikj} - X_{ik,-j}^\top \gamma_{kj}) (f_v(R_{ik}(\beta_0)^2) - X_{ik,-j}^\top \theta_{k, -j}^v)$, and  $\hat\Omega_{kj} = \frac{1}{n} \sum_{i =1}^n S_{ikj} S_{ikj}^\top$.

As an application of \citet[Theorem~1]{guo2025model}, we conclude without a duplicated proof that 
\begin{eqnarray}\label{eq: asymptotic distribution screening}
&&W_{kj}\stackrel{d}{\to} \chi_h^2 \mbox{ for } j\in \mA^c, \nonumber\\
&& W_{kj} \stackrel{d}{\to} \chi_h^2 (\delta_{kj} B_{kj}^\top \Omega_{kj} B_{kj}) \mbox{ for } j \in \mA,   
\end{eqnarray}
where $\delta_{kj} = \E [X_{kj}^2] - \E [X_{kj} X_{k, -j}^\top] \E[X_{k, -j}X_{k, -j}^\top] \E[X_{k, -j} X_{kj}]$, $\Omega_{kj} \coloneqq\lim_{n\to\infty} \hat\Omega_{kj}$, and $B_{kj}$ the $j$th row of the matrix $\bm B_k$.

With i.i.d.\ copies $X_{ik}$, $i \in [n]$, of $X_k \in \mathbbm{R}^p$, we estimate $\theta_k^v$ by solving the following penalized empirical risk minimization problem
\begin{equation}\label{eq: theta empirical minimization}
    \hat{\theta}_{k}^v
    =
    \argmin_{\theta \in \mathbbm{R}^p}
    \left\{
    \frac{1}{n}\sum_{i=1}^n
    \left(
    f_v(R_{ik}(\check\beta)^2) - X_{ik}^{\top}\theta
    \right)^2
    + \lambda \|\theta\|_1
    \right\}.
\end{equation}
Similarly, for each $j \in [p]$, we estimate $\gamma_{kj} \in \mathbbm{R}^{p-1}$ by solving
\begin{equation}\label{eq: gamma empirical minimization}
    \hat{\gamma}_{kj}
    =
    \argmin_{\gamma \in \mathbbm{R}^{p-1}}
    \left\{
    \frac{1}{n}\sum_{i=1}^n
    \left(
    X_{ikj} - X_{ik,-j}^{\top}\gamma
    \right)^2
    + \lambda \|\gamma\|_1
    \right\}.
\end{equation}
Under the above construction of the estimators $\hat{\theta}_k^v$ and $\hat{\gamma}_{kj}$, together with their convergence rates established in \eqref{eq: theta hat rate} and \eqref{eq: gammahat rate}, the test statistic defined in \eqref{eq: test statistic covariance active set selection} can be implemented in practice.

\subsubsection{Consistency of the plug-in estimators and Algorithm for the covariance function}\label{sssec: covariance function full estimation}
We establish in Theorem~\ref{thm: covariance matrix} the conditions for the consistency of the plug-in estimators $\hat\mA$ and $R(\check\beta)$. We then proceed to Algorithm~\ref{algorithm: variance function update} for implementation in practice, which results in nonparametric estimators of the covariance function $\bSigma(\bX_\mA; \beta_0)$ on the two disjoint subsamples for cross-fitting.

\begin{theorem}[Consistency of $R_i(\check\beta)$ and $\hat\mA$]\label{thm: covariance matrix}
Let $\check\beta$ be an initial estimator of $\beta_0$ satisfying 
\[
\|\check\beta - \beta_0\|_2 = O_P\left(\sqrt{\frac{s+m}{n}}\right).
\]
Suppose Assumptions~\ref{cond: fixed and bounded design}, \ref{cond: conditional mean function}, \ref{cond: residuals moments and subgaussian}, \ref{cond: sufficient dimension reduction}, then (i) $\max_{i,k}|R_{ik}(\check\beta) - R_{ik}(\beta_0)| = o_P(1)$; (ii) let $\alpha_p = p^{-c}$ and $c > 0$ some constant. It holds that $P(\hat \mA = \mA) \to 1$ for the critical value $t_0 = \chi_h^2(1 - \alpha_p/p)$, if the minimal signal strength satisfies $\min_{j \in \mA_k} \|B_{kj}\|_2 \gg  \sqrt{2(1 + c) \log p \max_{j} \frac{\lambda_{\max}(\Omega_{kj}^{-1})}{\delta_{kj}}}$, where $\delta_{kj}$ and $\Omega_{kj}$ are defined in \eqref{eq: asymptotic distribution screening}.
\end{theorem}

With the nonparametric estimator in \eqref{eq: nonparametric covariance function}, 
Theorem~\ref{thm: covariance matrix} guarantees that the plug-in estimators 
$\hat{\mA}$ and $R_i(\check\beta)$ consistently approximate their population 
counterparts $\mA$ and $R_i(\beta_0)$. Consequently, it is theoretically justified to proceed with Algorithm~\ref{algorithm: variance function update} using these plug-in quantities.

\begin{algorithm}[H]
\caption{Estimation of $\bSigma^{(q)} (x; \beta_0)$}
\label{algorithm: variance function update}
\begin{algorithmic}[1]
\REQUIRE Residuals $R_i (\check\beta^{(q)}) = Y_i - \bg(\bX_i^\top \check\beta^{(q)})$, $i \in \mathcal{I}_q$ with the $k$th element $R_{ik}(\check\beta)$.
\FOR{$k = 1, \ldots, l$}
    \STATE \textbf{(Screening)} For each $j = 1, \ldots, p$:
    \STATE Compute
   $\hat S_{ikj}^v = (X_{ikj} - X_{ik,-j}^\top \hat\gamma_{kj}) (f_v(R_{ik}(\check\beta)^2) - X_{ik,-j}^\top \hat\theta_{k, -j}^v), v = 1, \ldots, h $. 
    Let $\hat S_{ikj} = (\hat S_{ikj}^1,\cdots,\hat S_{ikj}^h)^{\top}\in \mathbb{R}^h$.
    \STATE Compute
    $\hat{S}_{kj} = \frac{1}{\sqrt{n}} \sum_{i=1}^n \hat S_{ikj}, 
    \hat\Omega_{kj} = \frac{1}{n} \sum_{i=1}^n \hat S_{ikj} \hat S_{ikj}^\top.
    $
    \STATE Compute the test statistic 
    $W_{kj} = \hat{S}_{kj}^\top \hat\Omega_{kj}^{-1} \hat{S}_{kj}$.
    \STATE \textbf{Active set:} Select
    \begin{eqnarray}\label{eq:active_set_screening}
        \hat\mA_k^{(q)} = \{ j : W_{kj} \ge \chi^2_h(1 - \alpha/p) \}.
    \end{eqnarray}
\ENDFOR
\STATE Define the overall active set as 
$\hat\mA^{(q)} = \cup_{k=1}^l \hat\mA_k^{(q)}$.

\STATE \textbf{Covariance estimation.} Construct $\hat\bSigma^{(q)}(x; \check\beta^{(q)})$ by \eqref{eq: nonparametric covariance function}.
\RETURN{$\hat\bSigma^{(q)}(x; \check\beta^{(q)})$.}
\end{algorithmic}
\end{algorithm}

\section{Asymptotic normality and Wald test}\label{sec: asymptotic normality}

We begin by establishing the oracle property of the solutions to \eqref{eq: partially penalized estimating equations}, which is defined using the true covariance function in Proposition~\ref{thm: oracle property}. 

\begin{proposition}\label{thm: oracle property}
    Under Assumptions~\ref{cond: fixed and bounded design}, \ref{cond: conditional mean function}, \ref{cond: residuals moments and subgaussian}, \ref{cond: eigenvalue max min finite}, \ref{cond: variance function bounded}, and \ref{cond: log p and n order}, and $s+m = o(\sqrt{n})$, then the following holds: (i) the solutions $\tilde\beta$ to \eqref{eq: partially penalized estimating equations} satisfy $\tilde\beta_{0, \left(\mMS\right)^c}  = \bm 0$. (ii) $\|\tilde\beta_{\mMS} - \beta_{0,\mMS}\|_2 = O_P(\sqrt{(s+m)/n})$. If $s + m = o(n^{1/3})$, then

\begin{eqnarray}\label{eq: asymptotic normal of subvector}
   \sqrt{n}\begin{pmatrix}
\tilde\beta_{\mM} - \beta_{0, \mM}  \\
  \tilde\beta_{\mS} - \beta_{0, \mS}
    \end{pmatrix} 
 \stackrel{d}{\to} 
   \mathcal{N}_{s+m} \left(\bm 0, \mathbf{V}_{1}^{-1} \left({\beta}_0 \right)\mathbf{V}_{2} \left({\beta}_0 \right)\mathbf{V}_{1}^{-1} \left(\beta_0 \right) \right), 
\end{eqnarray}
where 
${\mathbf{V}}_{1}(\beta_0) = \E \left[\bX_{i,\mMS} \bD_i(\beta_0) \bSigma\left(\bX_{i,\mA}\right)^{-1} \bD_i \left(\beta_0 \right)\bX_{i, \mMS}^{\top} \right]$ and 
\begin{eqnarray*}
\mathbf{V}_{2}(\beta_0) = \E \left[\bX_{i, \mMS} \bD_i(\beta_0) \bSigma\left(\bX_{i,\mA}\right)^{-1}
\E\left[ R_i(\beta_0) R_i(\beta_0)^\top \mid \bX_i \right]
 \bSigma\left(\bX_{i,\mA}\right)^{-1} \bD_i(\beta_0)\bX_{i, \mMS}^{\top} \right],
\end{eqnarray*}
where $R_i(\beta_0) = Y_i - \bg(\bX_i^\top \beta_0)$.
\end{proposition}

As in Proposition~\ref{prop: consistency}, Proposition~\ref{thm: oracle property} is stated for $\tilde\beta$ as the solution to the oracle estimating equations constructed using the true covariance function. However, the asymptotic normality result remains valid provided that the working covariance function used as a proxy is constructed independently of the dataset employed to estimate $\beta_0$. Thus, the asymptotic distribution result can be straightforwardly extended to $\check\beta$, the initial estimator with a working specification. 

Before presenting the final result for the cross-fitted estimator in
Section~\ref{ssec: cross-fitted efficiency}, which is based on refitted estimators obtained on subsamples using $\hat{\bSigma}(\cdot)$ estimated from the complementary subsamples, we first investigate in Lemma~\ref{lem: approximation error} whether $\hat U_n(\beta)=\frac{1}{n}\sum_{i=1}^n
\mathbf{X}_i \bD_i(\beta)
\hat{\bSigma}\left(
\mathbf{X}_{i,\mA}
\right)^{-1}
\bigl\{ Y_i - g(\bX_i^\top \beta) \bigr\}$ 
serves as an adequate proxy for its oracle counterpart $U_n$, given \eqref{eq: kernel estimator rate}.  

\subsection{Efficiency gain of the cross-fitted estimator with estimated covariance function $\hat\bSigma^{(q)}$}\label{ssec: cross-fitted efficiency}
Equipped with the consistent covariance function estimators
$\hat{\bSigma}^{(q)}(\cdot)$, $q \in \{1,2\}$, obtained from
Algorithm~\ref{algorithm: variance function update}, we construct a
cross-fitted estimator $\hat\beta$ defined in
\eqref{eq: two stage estimator}. For each split $q \in \{1,2\}$, let
$q' = 3 - q$ denote the complementary subsample index. The estimator
$\hat\beta^{(q')}$ is defined as a solution to the regularized estimating equation
\begin{equation}\label{eq: two stage GEE}
\bm 0 \in \hat{U}^{(q')}(\beta) + \partial \rho_\lambda(\beta; \mathcal M),
\end{equation}
where
\begin{eqnarray}\label{eq: cross fitted estimating equations}
    \hat{U}^{(q')}(\beta)
=
\frac{1}{|\mathcal I_{q'}|}
\sum_{i \in \mathcal I_{q'}}
\bX_i \bD_i(\beta)\,
\hat{\bSigma}^{(q)}\!\left(
\bX_{i,\hat{\mA}^{(q)}};\check\beta^{(q)}
\right)^{-1}
\bigl\{Y_i - \bg(\bX_i^\top\beta)\bigr\}.
\end{eqnarray}

The estimating equation is evaluated on the subsample $\mathcal I_{q'}$,
while the nuisance components, e.g., the estimated active set
$\hat{\mA}^{(q)}$ and the covariance function $\hat{\bSigma}^{(q)}(\cdot)$, are
learned from the auxiliary subsample $\mathcal I_q$. Consequently,
conditional on the auxiliary sample, the weight matrix
$\hat{\bSigma}^{(q)}(\bX_{i,\hat{\mA}^{(q)}};\check\beta^{(q)})$ is
independent of the response $Y_i$ for all $i \in \mathcal I_{q'}$.

This sample-splitting construction decouples the estimation error of the nuisance components from the model noise and thereby facilitates the asymptotic analysis. We next show that, under suitable regularity conditions, the resulting cross-fitted estimator $\hat{\beta}$ attains the same asymptotic behavior as the oracle estimator $\tilde{\beta}$ in Proposition~\ref{thm: oracle property}. The proof is deferred to Appendix~\ref{ssec: proof of Theorem 2}.

Specifically, we first show that the difference between the estimated estimating equation $\hat{U}_n$ and the oracle estimating equation $U_n$ is asymptotically negligible in Lemma~\ref{lem: approximation error}. Consequently, solving $\hat{U}_n$ yields estimators that are asymptotically equivalent to those obtained by solving $U_n$ with the true covariance structure, and the aggregated cross-fitted estimator $\hat{\beta}$ achieves near-oracle performance.


\begin{lemma}\label{lem: approximation error}
     Under Assumptions~\ref{cond: conditional mean function}--\ref{cond: sufficient dimension reduction}, we have \[
\sup_{\beta}
\|\hat{U}_n(\beta)-U_n(\beta)\|
= o_P(n^{-1/2}).
\]
Consequently, any solution $\hat{\beta}'$ to $ \bm 0 \in \hat{U}_n^{\rm p}(\beta) = \hat U_n (\beta) + \partial \rho_\lambda(\beta; \mM)$, satisfies
\[
\|\hat{\beta}'-\beta_0\|_2=o_P(1),
\]
where $\beta_0$ is the true parameter of interest.
\end{lemma}

\begin{theorem}\label{thm: estimator property}
    Under Assumptions~\ref{cond: conditional mean function}--\ref{cond: sufficient dimension reduction} and $s+m = o(\sqrt{n})$, then the following holds: (i) the cross-fitted estimator $\hat\beta$ satisfies $\hat \beta_{\left(\mMS\right)^c}  = \bm 0$. (ii) $\|\hat\beta_{\mMS} - \beta_{0,\mMS}\|_2 = O_P(\sqrt{(s+m)/n})$. If $s + m = o(n^{1/3})$, then

\begin{eqnarray}\label{eq: asymptotic normal of estimator subvector}
   \sqrt{n}\begin{pmatrix}
\hat\beta_{\mM} - \beta_{0, \mM}  \\
  \hat\beta_{\mS} - \beta_{0, \mS}
    \end{pmatrix} 
 \stackrel{d}{\to} 
   \mathcal{N}_{s+m} \left(\bm 0, \mathbf{V}_{1}^{-1}(\beta_0) \mathbf{V}_2 (\beta_0) \mathbf{V}_{1}^{-1}(\beta_0) \right), \nonumber
\end{eqnarray}
where the asymptotic variance is near oracle in \eqref{eq: asymptotic normal of subvector}.
\end{theorem}


\subsection{Test statistic}

We now form a Wald test statistic for the linear hypothesis involving the subvector $\beta_{\mathcal{M}}$, where $\mathcal{M}$ denotes the index set associated with the covariates that are of primary interest for testing. We consider the null hypothesis
\[
H_0: \bm{C} \beta_{0, \mathcal{M}} = t,
\]
against the local alternative
\[
H_{1,n}: \bm{C} \beta_{0, \mathcal{M}} = t + \frac{h}{\sqrt{n}},
\]
where $\bm{C} \in \mathbb{R}^{r \times m}$ is a fixed full-rank matrix, $t \in \mathbb{R}^r$ is a known vector, and $h \in \mathbb{R}^r$ is a local drift vector.

Using the asymptotic normality established in the previous corollary, we define the Wald test statistic based on the cross-fitted estimator $\hat{\beta}$ as:
\[
W_n = n (\bm{C} \hat{\beta}_{\mathcal{M}} - t)^\top \left( \bm{C} \hat{\mathbf{\Omega}} \bm{C}^\top \right)^{-1} (\bm{C} \hat{\beta}_{\mathcal{M}} - t),
\]
where $\hat{\mathbf{\Omega}}$ is a consistent estimator of the asymptotic covariance matrix of $\hat{\beta}_{\mathcal{M}}$.

Under the local alternative $H_{1,n}$, the statistic $W_n$ converges in distribution to a non-central chi-squared distribution, that is
\[
W_n \overset{d}{\to} \chi^2_r(\hat{\delta}),
\]
with non-centrality parameter $\hat{\delta} = h^\top \left( \bm{C} \hat{\bOmega} \bm{C}^\top \right)^{-1} h$.

The following theorem establishes that the test based on the combining estimator $\hat{\beta}$ achieves greater or equal power compared to that based on a standard (initial) estimator $\check{\beta}$, by comparing their respective non-centrality parameters.

\begin{theorem}[Power Improvement]\label{thm:power-improvement}
Under Assumptions~\ref{cond: conditional mean function}--\ref{cond: sufficient dimension reduction} and the local alternative hypothesis $H_{1,n}$, let $\hat{\delta}$ and $\check{\delta}$ denote the non-centrality parameters corresponding to the cross-fitted estimator $\hat{\beta}$ in \eqref{eq: two stage estimator} and the initial estimator $\check{\beta}$, respectively. Then, the local power of the Wald test based on $\hat{\beta}$ is always greater than or equal to that of the test based on $\check{\beta}$, that is,
\[
\hat{\delta} = h^\top \left( \bm{C} \hat\bOmega \bm{C}^\top \right)^{-1} h \geq h^\top \left( \bm{C} \check{\bOmega} \bm{C}^\top \right)^{-1} h = \check{\delta},
\]
where $\hat\bOmega$ and $\tilde\bOmega$ are the asymptotic covariance matrices associated with $\hat{\beta}$ and $\check{\beta}$, respectively.
\end{theorem}

Theorem~\ref{thm:power-improvement} demonstrates that, when additional information about the covariance structure is leveraged, the efficiency of the cross-fitted estimator $\hat{\beta}$ improves relative to the initial estimator $\check{\beta}$, resulting in a Wald test with enhanced power.


\section*{Appendix} \label{sec:appendix}
\appendix
\section{Assumptions}
For completeness, we summarize the assumptions used in our theory.
\begin{enumerate}[label=(A{{\arabic*}})]
  \item \label{cond: conditional mean function} Assume that  
  $\frac{1}{n} \sum_{i=1}^n\{ Y_{ik} - g(X_{ik}^\top \beta_0)\} = O_P(n^{-1/2})$. Further, $g (\cdot)$ is Lipschitz continuous, that is, there exists some constant $c_L$ such that 
  $|g(\mu_1) - g(\mu_2)| \leq c_L |\mu_1 - \mu_2|.$ In addition, the first, second, and third derivatives of $g$, denoted by $\dot g$, $\ddot g$, and $\dddot g$, respectively, exist and are bounded, that is, $|\dot{g} (\mu)|\leq c_1$, $|\ddot{g} (\mu) | < c_2$, and $|\dddot g (\mu)| < c_3$. 
  \item \label{cond: residuals moments and subgaussian}  Let $R_{i} (\beta) =Y_{i}-\bm g\left(\bX_{i}^{\top} \beta \right) \in \mathbb{R}^{l}$ be the $i$th residual associated with an unknown parameter vector $\beta$. For all $i \in [n]$ and $k \in [l]$, it holds that $\E[R_{ik}(\beta_0)] = 0$, where $\beta_0$ denotes the truth. Further, $R_{ik}(\beta_0)$ is sub-Gaussian, i.e., there exists $c_4 > 0$, such that $\E [\exp\{t R_{ik}(\beta_0)\}] \leq \exp(c_4 t^2), \forall t \in \mathbbm R$.
  \item \label{cond: variance function bounded} Let the diagonal elements of the variance function satisfy $\|\BSIGMA_{k, k}^{-1} (x)\|_\infty = \sup_{x} |\BSIGMA_{k,k}^{-1} (x)| <\infty, k\in [l]$. Further, the covariance function $\bSigma(\cdot)$ satisfies
$\sup_x \|\bSigma(x)^{-1}\|_{\rm op}
= \sup_x \lambda_{\max}\!\big(\bSigma(x)^{-1}\big)
\le c_5$.

  \item \label{cond: fixed and bounded design} For $k \in [l]$, $X_{ik} : i \in [n]$ is an i.i.d. sample where $\mathbb{E}[X_{ik}] = \bm{0}$ and $\bSigma_{X_k} = \E [X_{ik} X_{ik}^\top]$. Additionally, $X_{ik}$ is uniformly bounded and elliptical; that is, there exists a constant $c_6 > 0$ such that 
  $\|X_{ik}\|\leq c_6$ almost surely.
  \item \label{cond: active set and non active set covariance} For $i\in [n]$ and $k \in [l]$, it holds that $\max_{j \in \mMS, j'\in \mS^c} \E [X_{ikj} X_{ikj'}] = o(1)$, and $\max_{j\in \mS, j'\in \mS^c}\E [|X_{ikj} X_{ikj'} |] = O (1)$.
  \item \label{cond: eigenvalue max min finite} For $k \in [l]$, there are constants $c_{1k}$ and $c_{2k}$, such that with high probability, 
  \begin{eqnarray*}
      \lefteqn{0 < c_{1k} \leq \lambda_{k,\min} (\frac{1}{n} \sum_{i = 1}^n X_{ik, \mMS} X_{ik,\mMS}^\top)}&\\
      & \qquad\qquad\leq \lambda_{k,\max} (\frac{1}{n} \sum_{i = 1}^n X_{ik, \mMS} X_{ik, \mMS}^\top) \leq c_{2k} < \infty .
  \end{eqnarray*}
  \item \label{cond: log p and n order} $\log p = o\left\{n \lambda_n^{2} /(\log n)^{2}\right\}$ and $n \lambda_n^{2} /(\log n)^{2} \to \infty$. ${p(m+s)^{5/2}(\log n)^2}/{n^{3/2}\lambda_n } =o(1)$. ${p(m+s)^{5/2} (\log n)^3}/{n^{3/2}\lambda_n } = o(1)$.
  \item \label{cond: active set covariance matrix} There exists an active set $\mA \subseteq [p]$ that is common across $X_{ik}$, $k\in [l]$, such that the conditional covariance function depends only on the active covariates, as specified in \eqref{eq: conditional covariance matrix}.
  \item \label{cond: smoothness of covariance function} The covariance function $\bSigma(x; \beta_0)$ is H\"older-continuous of order $\nu \in (0, 1]$ in $x \in \mathcal X_{\mA} \subset \mathbb{R}^{l|\mA|}$.
  \item \label{cond: design variance} Assume that $\operatorname{vec}(\bX_{i,\mA})$ has a density $f_{\mA}$ supported on a compact set $\mathcal X_{\mA}$, and there exist constants $0<c_-<c_+<\infty$ such that
  \[
    c_- \le f_{\mA}(z) \le c_+, \quad \forall z \in \mathcal X_{\mA}.
  \]
  This condition guarantees stable kernel normalization and prevents degeneracy in covariance-function estimation.
  \item \label{cond: nonconvex regularizer} The penalty function $\rho_\lambda(t)$ is increasing and concave in $t\in [0, \infty)$
   \begin{enumerate}[label=(R\arabic*)]
    \item it has a continuous derivative $\dot\rho_\lambda(t)$ with $\dot\rho_\lambda(0^+) > 0$; 
    \item\label{cond: large value no thresholding} for some constant $a > 1$, $\dot\rho_\lambda(t) = 0, \forall t \ge a\lambda$;
    \item the derivative $\dot\rho_\lambda(t)$ is nonincreasing in $t$, and the second derivative $\ddot\rho_\lambda(t) \le 0, \forall\, t > 0$.
\end{enumerate}
  \item \label{cond: sufficient dimension reduction} For $k \in [l]$, we assume that
  \begin{enumerate}
      \item $\lambda_{\min}(\Omega_{kj}^{-1}) \ge c_{3k} > 0$ and $\lambda_{\max} (\Omega_{kj}^{-1}) \le c_{4k}$ uniformly in $j \in [p]$.
      \item $\delta_{kj} \ge c_{5k} > 0$ uniformly for $j \in [p]$.
  \end{enumerate}
\end{enumerate}

\section{Useful facts}


\begin{definition}{(Locally Lipschitz function).}
    A function $f: S \to \mathbbm R$ is locally Lipschitz if for any bounded $S \subseteq \mathbbm R^m$ with some generic dimension $m$.
\end{definition}

\begin{definition}[(Weak convexity)]\label{def: weak convex}
    For some $\mu > 0$, a locally Lipschitz function $\rho: \mathbbm R^m \to \mathbbm R$ is $\mu$-weakly convex if $\rho(v) + \frac{\mu}{2}\|v\|_2^2$ is convex. 
\end{definition}

\begin{definition}[Matrix H\"older Smoothness]\label{def: matrix smoothness}
Let $\bSigma(x) \in \mathbb{R}^{l \times l}$ be a matrix-valued function defined on $x \in \mathbb{R}^d$. We say that $\bSigma(x)$ is Hölder continuous of order $\nu \in (0,1]$ if there exists a constant $C$ such that
\[
\|\bSigma(x) - \bSigma(x')\|_F \leq C \|x - x'\|^\nu, \quad \forall x, x' \in \mathbb{R}^d,
\]
where $\|\cdot \|_F$ is the Frobenius norm.
\end{definition}

\begin{lemma}[Bernstein inequalities]\label{lemma: bernstein}
Let $Z_1, \ldots, Z_n$ be independent zero-mean random variables. \\
(i)
Suppose that $|Z_i| \leq M$ almost surely, for all $1 \leq i \leq n$ and some positive constant $M$. Then, for all $t > 0$,
\[
P\left( \left|\sum_{i=1}^n Z_i \right| > t \right) \leq 2 \exp \left\{- \frac{2^{-1} t^2}{\sum_{i=1}^n\E Z_i^2 + 3^{-1} M t} \right\}.
\]
(ii) 
Assume that there exists a constant $c_i>0$ such that
\begin{equation*}
\E\big[\exp\{t Z_i\}\big]\le \exp\!\left(\frac{c_i t^2}{2}\right),
\qquad \forall\, t\in\mathbb{R}.
\end{equation*}
Then, for all $s>0$,
\[
\Pr\!\left(\left|\sum_{i=1}^n Z_i\right|>s\right)
\le 2\exp\!\left(-\frac{s^2}{2\sum_{i=1}^n c_i}\right).
\]
\end{lemma}

\section{Proofs}
\subsection{Proof of Proposition~\ref{prop: consistency}}\label{ssec: proof of proposition consistency}
The estimating equations we are interested in are defined as
\[U_{n}(\beta)=\frac{1}{n} \sum_{i=1}^{n} \bX_{i,\mA} \bm{D}_{i}(\beta) \bSigma\left(X_{i}\right)^{-1}\left\{Y_{i}-\bm g\left(\bX_{i}^\top \beta\right)\right\},\]
where $\bSigma\left(\bX_{i,\mA}\right) \in \mathbb{R}^{l \times l}, \bm D_{i}(\beta)=\mbox{diag}\left(\dot{g}\left(X_{i}^\top \beta\right), \cdots, \dot{g}\left(X_{i}^\top\beta\right)\right)$. The nonzero components of the true parameter vector $\beta_{0, \mS}$, where $\mS=\left\{j: \beta_{j} \neq 0\right\}$, have size $|\mS| = s$. The subvector of interest $\beta_{0, \mM}$, where $\mM = \left\{j: \beta_{j} \mbox{ unpenalized} \right\}$, has size $|\mM| = m$. Without loss of generality, we assume $\mM\cap \mS = \emptyset$. 

For $\tau > 0$, we define the following sets 
\begin{equation*}
\mathcal{N}_{\beta, \tau}=\left\{\beta \in \mathbb{R}^{p}:\left\|\beta_{\mMS}-\beta_{0, \mMS}\right\|_2 \leq \tau \sqrt{\frac{m+s}{n}},
\beta_{(\mMS)^c}=0\right\},
\end{equation*}
\begin{equation*}
\partial \mathcal{N}_{\beta, \tau}=\left\{\beta \in \mathbb{R}^{p}: \|\beta_{\mMS}-\beta_{0, \mMS} \|_2=\tau \sqrt{\frac{s+m}{n}}, \beta_{(\mMS)^c}=0 \right\}.
\end{equation*}

\subsubsection{Step 1: proof of \eqref{eq: prop active set estimation consistency}}
We show the existence of $\tilde{\beta}$, where $\tilde{\beta}_{\mMS}$ is the oracle solution and $\tilde{\beta}_{(\mMS)^c}=0$. We show that $\tilde{\beta}-\beta_{0}=O_{p}(\sqrt{(s+m)/ n})$. Then, following \citet{ortega2000iterative}, we show that for any $\varepsilon>0$, there exists a $\tau >0$, such that for all sufficiently large $n$, 
\[P\left\{\sup_{\beta\in \partial \mathcal{N}_{\beta, \tau}} \left(\beta-\beta_{0}\right)^{\top} U_n(\beta)<0\right\} \geqslant 1-\varepsilon.\] 

Since we constrain on $\partial \mathcal{N}_{\beta, \tau}$, we have \[\left(\beta-\beta_{0}\right)^{\top} U_n(\beta)=\left(\beta_{\mMS}-\beta_{0, \mMS}\right)^{\top} U_{\mMS}(\beta),\]
where $U_{\mMS}(\beta)$ is formed by the equations corresponding to the index set $\mMS$. Let $\nabla_{\mMS} U_{\mMS}(\beta) = \frac{\partial U_{n, \mMS}(\beta)}{\partial \beta_{\mMS}}$. By Taylor expansion,
\begin{eqnarray*}
   \lefteqn{\left(\beta_{\mMS} - \beta_{0, \mMS} \right)^{\top} U_{\mMS}(\beta)
 = \left(\beta_{\mMS}-\beta_{0, \mMS}\right)^{\top} U_{\mMS}\left(\beta_{0}\right) }& \\ 
&+ \left(\beta_{\mMS}-\beta_{0, \mMS}\right)^{\top} \nabla_{\mMS} U_{\mMS}\left(\beta^*\right)\left(\beta_{\mMS}-\beta_{0, \mMS}\right) \coloneqq T_1 +T_2 ,
\end{eqnarray*}
where $\beta^*$ satisfies that $\beta^*_{\mMS}$ is between $\beta_{0, \mMS}$ and $\tilde\beta_{\mMS}$ and $\beta^*_{(\mMS)^c} = 0_{m+s}$ is a null vector. 
Consider $T_{1}$, by Cauchy-Schwarz, 
\begin{equation}\label{eq: T1 bounded}
\left|T_{1}\right| \leq \tau \sqrt{(s+m)/{n}}\left\|U_{\mMS}\left(\beta_{0}\right)\right\|_2.
\end{equation}
Moreover,
\begin{eqnarray}\label{eq: estimating equation at true bounded}
\lefteqn{\E\left\{\left\|U_{\mMS}\left(\beta_{0}\right)\right\|^{2}_2\right\} 
} && \\
& = & \E\left\{\left\|\frac{1}{n} \sum\limits_{i=1}^n\bX_{i, \mMS} \bD\left(\beta_0\right)\bSigma\left(\bX_{i,\mA}\right)^{-1}\left\{Y_i - \bg\left(\bX_i^{\top}\beta_0\right)\right\}\right\|_2^2\right\}  \nonumber\\
& = & \E\left\{\frac{1}{n^2} \sum\limits_{i=1}^n\left\|\bX_{i, \mMS} \bD\left(\beta_0\right)\bSigma\left(\bX_{i,\mA}\right)^{-1}\left\{Y_i - \bg\left(\bX_i^{\top}\beta_0\right)\right\}\right\|_2^2\right\} \nonumber\\
& = & \mbox{tr}\bigg\{\E\Big\{\frac{1}{n^2} \sum\limits_{i=1}^n\bX_{i, \mMS} \bD\left(\beta_0\right)\bSigma\left(\bX_{i,\mA}\right)^{-1}\left\{Y_i - \bg\left(\bX_i^{\top}\beta_0\right)\right\}  \nonumber\\
&& \qquad\qquad \left\{Y_i - \bg\left(\bX_i^{\top}\beta_0\right)\right\}^{\top}\bSigma\left(\bX_{i,\mA}\right)^{-1}\bD\left(\beta_0\right)\bX_{i, \mMS}^{\top}\Big\} \bigg\}   \nonumber\\
& = &  \frac{1}{n^2} \mbox{tr}\left\{\sum\limits_{i=1}^n\bX_{i, \mMS} \bD\left(\beta_0\right)\bSigma\left(\bX_{i,\mA}\right)^{-1}\bD\left(\beta_0\right)\bX_{i, \mMS}^{\top} \right\}\nonumber\\
&\leq & \frac{C}{n^2} \mbox{tr}\{\sum\limits_{i=1}^n \bX_{i, \mMS} \bX^\top_{i, \mMS} \} \nonumber\\
&= &  \frac{C}{n^2} \sum\limits_{i=1}^n  X^\top_{ik,\mMS}X_{ik, \mMS}  \nonumber\\
&\leq & \frac{C}{n^2} \sum\limits_{i=1}^n\sum\limits_{k=1}^ l(s+m)
\nonumber\\
&=& O(\frac{m+s}{n})   \nonumber.
\end{eqnarray}
where the first inequality holds by assumptions~\ref{cond: conditional mean function}, \ref{cond: variance function bounded}; the second inequality holds by Assumption~\ref{cond: fixed and bounded design}; the last equality holds since $l$ is fixed. Then, combining \eqref{eq: T1 bounded} and \eqref{eq: estimating equation at true bounded}, we conclude that $\left|T_{1}\right|<\tau O(\frac{m+s}{n})$.

Then, we consider $T_{2}$ 
\begin{eqnarray*}
\lefteqn{T_{2}=\left(\beta_{\mMS}-\beta_{0, \mMS}\right)^{\top} \nabla_{\mMS} U_{\mMS}\left(\beta^*\right)\left(\beta_{\mMS}-\beta_{0, \mMS}\right)}&& \\
&=&\left(\beta_{\mMS}-\beta_{0, \mMS}\right)\nabla_{\mMS} U_{\mMS}\left(\beta_{0}\right)\left(\beta_{\mMS}-\beta_{0, \mMS}\right) \\
&+& \left(\beta_{\mMS}-\beta_{0, \mMS}\right)^{\top}\left\{\nabla_{\mMS} U_{\mMS} (\beta^*)-\nabla_{\mMS} U_{\mMS}\left(\beta_{0}\right)\right\}\left(\beta_{\mMS}-\beta_{0, \mMS}\right) \\
&=& T_{21} + T_{22},
\end{eqnarray*}
where 
\begin{eqnarray*}
\lefteqn{\nabla_{\mMS} U_{\mMS}(\beta)  
= \frac{1}{n} \sum_{i=1}^n \sum_{k = 1}^l }& \\
&X_{ik, \mMS} X_{ik, \mMS}^\top \left\{\ddot g(X_{ik}^\top \beta) \left[\bSigma(\bX_{i,\mA})^{-1}\right]_{kk} \left\{Y_{ik} 
- g(X_{ik}^\top \beta) \right\}
 -  \dot g(X_{ik}^\top \beta)^2 \left[\bSigma\left(\bX_{i,\mA}\right)^{-1}\right]_{kk} \right\}
\end{eqnarray*}
For $T_{21}$, 
\begin{eqnarray*}
    \lefteqn{T_{21} = \left(\beta_{\mMS}-\beta_{0, \mMS}\right)\nabla_{\mMS} U_{\mMS}\left(\beta_{0}\right)\left(\beta_{\mMS}-\beta_{0, \mMS}\right) } &\\
    &= - \frac{1}{n} \sum_{i=1}^n (\beta_{\mMS} - \beta_{0, \mMS})^\top \bX_{i, \mMS} \bX^\top_{i, \mMS} (\beta_{\mMS} - \beta_{0, \mMS})^\top \{ \bD_i(\beta_0) \bSigma\left(\bX_{i,\mA}\right)^{-1} \bD_i^\top(\beta_0) \\
   & - \sum_{k = 1}^l \ddot g(X_{ik}^\top \beta_0) \left[\bSigma\left(\bX_{i,\mA}\right)^{-1}\right]_{kk} \{Y_{ik} - g(X_{ik}^\top \beta_0)\} \},
\end{eqnarray*}
where the second term is of order $O_P(\sqrt{1/n})$ by Assumption~\ref{cond: conditional mean function} and is negligible when $n \to \infty$. Thus, we only focus on the first term of $T_{21}$. By Cauchy-Schwarz inequality and Assumptions~\ref{cond: conditional mean function}, \ref{cond: variance function bounded}, \ref{cond: eigenvalue max min finite}, 
\begin{eqnarray*}
  T_{21} &\leq& - C \|\beta_{\mMS} - \beta_{0, \mMS}\|^2_2 \max_k \lambda_{k, \max} (\frac{1}{n}\sum_i X_{ik} X_{ik}^\top) + O_P(\sqrt{1/n})\\
  &\leq& (-C + O_P(\sqrt{1/n})) \tau^2 (s+m)/n,
\end{eqnarray*}
where the second inequality holds since $\beta$ is constrained on $\partial \mathcal{N}_{\beta, \tau}$ for some $\tau >0$.

We next show that $T_{22}$ is upper bounded with the leading term being negative. Notice that
\begin{eqnarray*}
\lefteqn{\nabla_{\mMS} U_{\mMS}(\beta^*)-\nabla_{\mMS} U_{\mMS}\left(\beta_{0}\right)}&&\\
&=& -\frac{1}{n}\sum_{i=1}^{n} \bX_{i, \mMS} \bX_{i, \mMS}^{\top}\{ \bm D_{i}(\beta^*) \bSigma\left(\bX_{i,\mA}\right)^{-1} \bm D_{i}^{\top}(\beta^*) - \bm D_{i}\left(\beta_{0}\right)\bSigma\left(\bX_{i,\mA}\right)^{-1}  \bm D_{i}^{\top}\left(\beta_{0}\right) \\
&& - \sum_{k=1}^l \{\ddot g(X_{ik}^\top \beta^*)[\bSigma\left(\bX_{i,\mA}\right)^{-1}]_{kk}\{Y_{ik} - g(X_{ik}^\top \beta^*)\} - \ddot g(X_{ik}^\top \beta_0) [\bSigma\left(\bX_{i,\mA}\right)^{-1}]_{kk}\{Y_{ik} - g(X_{ik}^\top \beta_0)\} \} \\
&=&  -\frac{1}{n}\sum_{i=1}^{n} \bX_{i, \mMS} \bX_{i, \mMS}^{\top}\{ \left(\bm D_{i}\left(\beta^* \right) - \bD_i \left(\beta_0 \right )\right) \bSigma\left(\bX_{i,\mA}\right)^{-1} \left(\bm D_{i}\left(\beta^* \right) - \bm D_{i}\left(\beta_{0}\right)\right)^\top \\
&& - \sum_{k=1}^l \{\ddot g(X_{ik}^\top \beta^*)[\bSigma\left(\bX_{i,\mA}\right)^{-1}]_{kk}\{Y_{ik} - g(X_{ik}^\top \beta^*)\} - \ddot g(X_{ik}^\top \beta^*) [\bSigma\left(\bX_{i,\mA}\right)^{-1}]_{kk}\{Y_{ik} - g(X_{ik}^\top \beta_0)\}\}  \\
&& + \ddot g(X_{ik}^\top \beta^*)[\bSigma\left(\bX_{i,\mA}\right)^{-1}]_{kk}\{Y_{ik} - g(X_{ik}^\top \beta_0)\} - \ddot g(X_{ik}^\top \beta_0) [\bSigma\left(\bX_{i,\mA}\right)^{-1}]_{kk}\{Y_{ik} - g(X_{ik}^\top \beta_0)\} \} \\
 &=&-\frac{1}{n}\sum_{i=1}^{n} \bX_{i, \mMS} \bX_{i, \mMS}^{\top}\{ \left(\bm D_{i}\left(\beta^* \right) - \bD_i \left(\beta_0 \right )\right) \bSigma\left(\bX_{i,\mA}\right)^{-1} \left(\bm D_{i}\left(\beta^* \right) - \bm D_{i}\left(\beta_{0}\right)\right)^\top \\
&& - \sum_{k=1}^l \{\ddot g(X_{ik}^\top \beta^*)[\bSigma\left(\bX_{i,\mA}\right)^{-1}]_{kk}\{g(X_{ik}^\top \beta_0) - g(X_{ik}^\top \beta^*)\} \}  \\
&& \qquad\qquad\qquad + (\ddot g(X_{ik}^\top \beta^*) -\ddot g(X_{ik}^\top \beta_0))[\bSigma\left(\bX_{i,\mA}\right)^{-1}]_{kk}\{Y_{ik} - g(X_{ik}^\top \beta_0)\} \}.
\end{eqnarray*}

Then, we conclude that 
\begin{eqnarray*}
   T_{22} &=& \left(\beta_{\mMS}-\beta_{0, \mMS}\right)^{\top}\left\{\nabla_{\mMS} U_{\mMS} (\beta^*)-\nabla_{\mMS} U_{\mMS}\left(\beta_{0}\right)\right\}\left(\beta_{\mMS}-\beta_{0, \mMS}\right)\\
    &\leq & (- C + O_P (\sqrt{1/n}))\tau^2 (m+s)/n,
\end{eqnarray*}
by Assumptions~\ref{cond: conditional mean function}, \ref{cond: variance function bounded}, \ref{cond: eigenvalue max min finite} and $\beta^* \in \mathcal N$ satisfying $\|\beta_{\mMS}^* - \beta_{0, \mMS}\|_2 \leq \tau \sqrt{(m+s)/n} $.
Hence, 
$$
T_{2}=T_{21}+T_{22} \leq  \left(- C + O_P(\sqrt{1/n})\right)\tau^{2} \frac{m+s}{n}.
$$

Then, for sufficiently large $\tau$, $T_2$ dominates $\left(\beta-\beta_{0}\right)^{\top} U_n(\beta)$
and is negative for all sufficiently large $n$. Therefore, $\|\tilde{\beta}-\beta_0\|_2 =O_P\left(\sqrt{\frac{m+s}{n}}\right)$ and by the construction of $\tilde\beta$, \eqref{eq: prop active set estimation consistency} holds.

\subsection{Step~2: proof of \eqref{eq: equations negligible 1}}
By the construction of $\tilde\beta$, for $j \in \mMS$, we have $U_j(\tilde{\beta}) = 0$. Thus, it suffices to show that $P\left(\dot\rho_\lambda (\beta_j; \mM) = 0, j \in \mMS \right)= P\left(\dot\rho_\lambda (\beta_j) = 0 , j \in \mS \right) \to 1$. By Assumption~\ref{cond: large value no thresholding}, it is equivalent to show that $P\left( |\tilde\beta_j| \geq a \lambda, j \in \mS \right)$. Note that 
\begin{eqnarray*}
    \min_{j \in \mS} |\tilde{\beta}_j | &=& \min_{j \in \mS} |\beta_{0, j} + \tilde\beta_j + \beta_{0,j}| \\
    &\geq& \min_{j \in \mS} |\beta_{0, j}| - \max_{j \in \mS} |\tilde{\beta}_j - \beta_{0, j}| \\
    &\geq& \min_{j \in \mS} |\beta_{0, j}| - \|\tilde \beta_{\mS} - \beta_{0, \mS}\|_2.
\end{eqnarray*}
Therefore, we have 
\begin{eqnarray*}
    P\left\{ \left( \min_{j \in \mS} |\beta_{0, j}| - \|\tilde \beta_{\mS} - \beta_{0, \mS}\|_2 \right) \geq a \lambda \right\} = P \left\{ \|\tilde \beta_{\mS} - \beta_{0, \mS}\|_2 \leq \left( \min_{j \in \mS} |\beta_{0, j}|  - a \lambda\right) \right\} \to 1,
\end{eqnarray*}
as $\min_{j \in \mS} |\beta_{0, j}|/\lambda \to \infty$ and $\|\tilde\beta_{\mS} - \beta_{0, \mS}\|_2 \leq \|\tilde\beta_{\mMS} - \beta_{0, \mMS}\| = o(\lambda)$. Thus, $P\left(\min_{j \in \mS} |\tilde \beta_j| \geq a\lambda \right) \to 1$ for $n \to \infty$.

\subsection{Step~3: proof of \eqref{eq: equations negligible 2}}
We show that  
\begin{equation*}
P\left\{\left|U_{n,j}^{\rm P}(\tilde{\beta})\right| \leq \frac{\lambda_n}{\log n}, j \in(\mMS )^{c}\right\} \to 1 .  
\end{equation*}
By construction, for $j \in (\mMS)^c$, $\tilde{\beta}_{j}=0$. 
Further, $\dot\rho_\lambda(\tilde\beta_j) \in [-\lambda_n, \lambda_n]$ for the nonconvex penalties considered in this project. For $\dot\rho_\lambda(\tilde\beta_j)$ sufficiently small, it suffices to show that
\begin{eqnarray}\label{eq: prop 1 proof step 3 equivalent}
   P\left\{\max_{j \in (\mMS)^{c}}\left|U_{n,j}(\tilde{\beta} )\right| \leq \frac{\lambda_n}{\log n}\right\} \rightarrow 1. 
\end{eqnarray}
By Taylor expansion, we obtain 
\begin{eqnarray*}
    U_{n, j}(\tilde\beta) = U_{n, j}(\beta_0) + \frac{\partial U_{n,j}(\beta_0)}{\partial \beta} (\tilde\beta - \beta_0) + \frac{1}{2}(\tilde\beta - \beta_0)^\top \frac{\partial^2 U_{n,j}(\beta^*)}{\partial \beta \partial \beta^\top}(\tilde\beta - \beta_0),
\end{eqnarray*}
where $\beta^*$ is between $\tilde\beta$ and $\beta_0$. Let $\nabla_{\mMS} U_{n, j}(\beta) =  \frac{\partial U_{n,j}(\beta)}{\partial \beta_{\mMS}}$ and $\bm H_{j, \mMS} (\beta) = \frac{\partial^2 U_{n,j}(\beta)}{\partial \beta_{\mMS} \partial \beta_{\mMS}^\top}$. Since $\tilde\beta_{(\mMS)^c} = 0$, we have 
\begin{eqnarray}\label{eq: taylor expansion U_nj}
    U_{n,j}(\tilde\beta) &=& U_{n, j}(\beta_0) + \nabla_{\mMS}U_{n, j}(\beta_0) (\tilde\beta_{\mMS} - \beta_{0, \mMS}) \nonumber \\
    &&+ \frac{1}{2}(\tilde\beta_{\mMS} - \beta_{0, \mMS})^\top \bm H_{j, \mMS} (\beta) (\tilde\beta_{\mMS} - \beta_{0, \mMS}).
\end{eqnarray}
We then have
\begin{eqnarray*}
    \lefteqn{P\left\{\max _{j \in (\mMS)^{c}}\left|U_{n,j}(\tilde{\beta})\right|>\frac{\lambda_n}{\log n}\right\}} &&\\
   &\leq& P\left\{\max_{j \in(\mMS)^{c}}\left|U_{n,j}\left(\beta_{0}\right)\right|>\frac{\lambda_n}{3 \log n}\right\} + P\left\{\max_{j \in(\mMS)^{c}}\left|\nabla_{\mMS}U_{n,j}(\beta_0)\left(\tilde{\beta}_{\mMS}-\beta_{0, \mMS}\right)\right| > \frac{\lambda_n}{3 \log n}\right\} \\
   &+& P\left\{\max_{j \in (\mMS)^c} \left|(\tilde\beta_{\mMS} - \beta_{0, \mMS})^\top \bm H_{j, \mMS} (\beta) (\tilde\beta_{\mMS} - \beta_{0, \mMS}) \right| > \frac{\lambda_n}{3 \log n} \right\}\\
   &\coloneqq& T_{3}+T_{4}+T_{5}
\end{eqnarray*}
Then, \eqref{eq: prop 1 proof step 3 equivalent} holds if $T_3, T_4, T_5 = o(1)$.
For $j \in [p]$,
\[
U_{n, j} (\beta) = \frac{1}{n} \sum_{i=1}^n X_{ij} \bD_i (\beta_0) \bSigma\left(\bX_{i,\mA}\right)^{-1}(Y_i - \bm g(\bX_i^\top\beta_0)).
\]
\begin{eqnarray}\label{eq: T3 1}
\lefteqn{T_{3}=P\left\{\max_{j \in(\mMS)^{c}}\left|U_{n,j}\left(\beta_{0}\right)\right|>\frac{\lambda_n}{3 \log n}\right\}}&& \\
&\leq& \sum_{j \in (\mMS)^c} P\left\{\left|\frac{1}{n} \sum_{i=1}^n X_{ij} \bD_i (\beta_0) \bSigma\left(\bX_{i,\mA}\right)^{-1}(Y_i - \bm g(\bX_i^\top\beta_0))\right|
>\frac{\lambda_n}{3 \log n}\right\}. \nonumber
\end{eqnarray}
Taking $Z_{i}= X_{ij} \bD_i (\beta_0) \bSigma\left(\bX_{i,\mA}\right)^{-1}(Y_i - \bm g(\bX_i^\top\beta_0))$, which is sub-Gaussian by Assumption~\ref{cond: residuals moments and subgaussian}. Further, we define a vector conditional on $\bX_i$ as $a_i := (X_{ij}\,\bD_i(\beta_0)\,\bSigma(\bX_{i,\mA})^{-1})^\top$, where $\|a_i\|_2$ is uniformly upper bounded in $i$ by Assumptions~\ref{cond: conditional mean function}, \ref{cond: variance function bounded}, \ref{cond: fixed and bounded design}. It holds that $Z_i = a_i^\top \bR_i(\beta_0)=\sum_{k=1}^{l} a_{ik} R_{ik}(\beta_0)$.

Then,
\begin{eqnarray*}
\E\left[\exp\{t Z_i\}\mid \bX_i\right]
&= \E\left[\exp\Big\{t\sum_{k=1}^l a_{ik}R_{ik}(\beta_0)\Big\}\,\Big|\,\bX_i\right] = \prod_{k=1}^{l} \E\left[\exp\{t a_{ik}R_{ik}(\beta_0)\}\mid \bX_i\right]\\ 
&\le \prod_{k=1}^{l} \exp\left(c_4 (t a_{ik})^2\right)
= \exp\left(c_4 t^2 \sum_{k=1}^l a_{ik}^2\right) 
= \exp\left(c_4 t^2 \|a_i\|_2^2\right).
\end{eqnarray*}

Then, by Lemma~\ref{lemma: bernstein} (ii), we obtain
\begin{eqnarray}\label{eq: T3 2}
\lefteqn{P\left\{\left|\frac{1}{n} \sum_{i=1}^n X_{ij} \bD_i(\beta_0) \bSigma\left(\bX_{i,\mA}\right)^{-1}(Y_i - \bm g(\bX_i^\top\beta_0))\right|
>\frac{\lambda_n}{3 \log n}\right\}} && \\
&\leq& 2 \exp \left\{-\frac{2^{-1} \left(\frac{n \lambda_n }{3 \log n}\right)^{2}}{\sum_{i=1}^{n} \E\left[Z_{i}^{2}\right]+3^{-1} M \frac{n\lambda_n}{3 \log n}}\right\} \nonumber\\
&\leq& 2 \cdot \exp \left\{-C \frac{\left(\frac{n \lambda_n}{3 \log n}\right)^{2}}{n}\right\} =2 \exp \left\{-C' n \left(\frac{\lambda_n}{\log n}\right)^{2} \right\}. \nonumber
\end{eqnarray}

Combine \eqref{eq: T3 1} and \eqref{eq: T3 2}, it holds that 
\begin{eqnarray*}
 \lefteqn{T_3 \leq \sum_{j \in(\mMS)^{c}} P\left\{\left|\frac{1}{n} \sum_{i=1}^n X_{ij} \bD_i (\beta_0) \bSigma\left(\bX_{i,\mA}\right)^{-1}(Y_i - \bm g(\bX_i^\top\beta_0))\right|
>\frac{\lambda}{3 \log n}\right\}}&&\\
 && \qquad\qquad\qquad \leq 2 \exp \left\{\log p-C' n\left(\frac{\lambda}{\log n}\right)^{2}\right\}=o(1),
 \end{eqnarray*}
where the equality holds by Assumption~\ref{cond: log p and n order}.

Second, we show that $T_{4}=o(1)$.
\begin{eqnarray}\label{eq: T4 1}
    \lefteqn{P\left\{\max_{j \in(\mMS)^{c}}\left|\nabla_{\mMS}U_{n,j}(\beta_0)\left(\tilde{\beta}_{\mMS}-\beta_{0, \mMS}\right)\right| > \frac{\lambda}{3 \log n}\right\}}&&\\
    &=& P\left\{\max_{j \in(\mMS)^{c}}\left|\nabla_{\mMS}U_{n,j}(\beta_0)\left(\tilde{\beta}_{\mMS}-\beta_{0, \mMS}\right)\right| > \frac{\lambda}{3 \log n}, \|\tilde\beta_{\mMS} - \beta_{0, \mMS}\|_2 \leq \sqrt{\frac{m+s}{n}} \log n\right\}\nonumber\\
 &+& P\left\{
 \|\tilde\beta_{\mMS} - \beta_{0, \mMS}\|_2 > \sqrt{\frac{m+s}{n}}\log n\right\} \nonumber\\
 &\leq& P\left\{\max_{j \in(\mMS)^{c}}\left\|\nabla_{\mMS}U_{n,j}(\beta_0)\right\|_2 > \frac{\lambda n^{3/2}}{3\sqrt{m+s} (\log n)^2} \right\} + o(1) \nonumber\\
 &\leq & \sum_{j\in (\mMS)^{c}}P\left\{\left \|\nabla_{\mMS}U_{n,j}(\beta_0)\right\|_2 > \frac{\lambda n^{3/2}}{3\sqrt{m+s} (\log n)^2} \right\} + o(1) \nonumber \\
 &\leq & \sum_{j\in (\mMS)^{c}} P\left\{\sum_{j'\in (\mMS)}\nabla_{j'}U_{n,j}(\beta_0)^2  > \frac{\lambda^2 n^3}{9 (m+s)(\log n)^4} \right\} + o(1), \nonumber
\end{eqnarray}
where the first inequality holds by the Cauchy-Schwarz inequality.

Recall that
\begin{eqnarray*}
\lefteqn{\nabla_{j'} U_{n,j}(\beta) }&& \\
&=& \frac{1}{n} \sum_{i=1}^n \sum_{k = 1}^l X_{ikj} X_{ikj'}^\top \left\{\ddot g(X_{ik}^\top \beta) \left[\bSigma\left(\bX_{i,\mA}\right)^{-1}\right]_{kk} \left\{Y_{ik} - g(X_{ik}^\top \beta) \right\}
 -  \dot g(X_{ik}^\top \beta)^2 \left[\bSigma\left(\bX_{i,\mA}\right)^{-1}\right]_{kk} \right\},
\end{eqnarray*}
and 
\begin{eqnarray*}
   \lefteqn{\E |\nabla_{j'} U_{n,j}(\beta_0)|} && \\
   &=& \E \left|\frac{1}{n} \sum_{i=1}^n \sum_{k = 1}^l X_{ikj} X_{ikj'} \left\{\ddot g(X_{ik}^\top \beta_0) \left[\bSigma\left(\bX_{i,\mA}\right)^{-1}\right]_{kk} \left\{Y_{ik} - g(X_{ik}^\top \beta_0) \right\}
 -  \dot g(X_{ik}^\top \beta_0)^2 \left[\bSigma\left(\bX_{i,\mA}\right)^{-1}\right]_{kk} \right\}\right| \\
 &=& \E \left|\frac{1}{n} \sum_{i=1}^n \sum_{k = 1}^l X_{ikj} X_{ikj'} \ddot g(X_{ik}^\top \beta_0) \left[\bSigma\left(\bX_{i,\mA}\right)^{-1}\right]_{kk} \left\{Y_{ik} - g(X_{ik}^\top \beta_0) \right\}\right| \\
 &&+  \left|\frac{1}{n} \sum_{i=1}^n \sum_{k = 1}^l X_{ikj} X_{ikj'}^\top  \dot g(X_{ik}^\top \beta_0)^2 \left[\bSigma\left(\bX_{i,\mA}\right)^{-1}\right]_{kk} \right| \\
 &\leq& \frac{C_1}{n} \sum_{i=1}^n \sum_{k = 1}^l \left|X_{ikj} X_{ikj'} \right| \E \left[\left|Y_{ik} - g(X_{ik}^\top \beta_0)\right|\right] + \frac{C_2}{n}\sum_{i=1}^n \sum_{k = 1}^l \left|X_{ikj} X_{ikj'} \right|\\
 &\leq& C.
\end{eqnarray*}
where the first inequality holds by Assumptions~\ref{cond: conditional mean function}, \ref{cond: variance function bounded}; the second inequality holds by Assumptions~\ref{cond: residuals moments and subgaussian} and \ref{cond: active set and non active set covariance}.


Continue on \eqref{eq: T4 1}, we conclude that
\begin{eqnarray*}
    \lefteqn{P\left\{\max_{j \in(\mMS)^{c}}\left|\nabla_{\mMS}U_{n,j}(\beta_0)\left(\tilde{\beta}_{\mMS}-\beta_{0, \mMS}\right)\right| > \frac{\lambda_n}{3 \log n}\right\} } &&\\
    &\leq& p  P\left\{ \sum_{j' \in \mMS}\left |\nabla_{j'}U_{n,j}(\beta_0)\right | >  \frac{\lambda_n n^{3/2}}{3 (m+s)^{3/2}(\log n)^2} \right\} + o(1) \\
    &\leq&  \E|\nabla_{j'}U_{n,j}(\beta_0)| \frac{3 (m+s)^{3/2}(\log n)^2}{n^{3/2}\lambda_n } \\
    &=& O( \frac{p(m+s)^{5/2}(\log n)^2}{n^{3/2}\lambda_n }) = o(1),
\end{eqnarray*}
where the last equality holds by Assumption~\ref{cond: log p and n order}.

Third, we show that $T_{5}= o(1)$. 
\begin{eqnarray}\label{eq: T5}
 \lefteqn{T_{5} = P\left\{\max_{j \in (\mMS)^c} \left|(\tilde\beta_{\mMS} - \beta_{0, \mMS})^\top \bm H_{j, \mMS} (\beta) (\tilde\beta_{\mMS} - \beta_{0, \mMS}) \right| > \frac{\lambda_n}{3 \log n} \right\}}&& \nonumber\\
 &\leq& P\{\max_{j \in (\mMS)^c} \left|(\tilde\beta_{\mMS} - \beta_{0, \mMS})^\top \bm H_{j, \mMS} (\beta) (\tilde\beta_{\mMS} - \beta_{0, \mMS}) \right| > \frac{\lambda_n}{3 \log n} , \nonumber\\
 &&\|\tilde\beta_{\mMS} - \beta_{0, \mMS}\|_2 \leq \sqrt{\frac{m+s}{n}} \log n\}\nonumber\\
 &&+ P\left\{\|\tilde\beta_{\mMS} - \beta_{0, \mMS}\|_2 > \sqrt{\frac{m+s}{n}} \log n\right\} \nonumber\\
 &\leq& \sum_{j \in (\mMS)^c} P\left\{ \max_{j'\in \mMS} |\bm H_{jj'} (\tilde\beta^*)\}  |  > \frac{n\lambda_n}{3 (m+s) (\log n)^3} \right\} + o(1)\nonumber \\
 &\leq & \sum_{j \in (\mMS)^c} \sum_{j' \in \mMS}  P\left\{ | \bm H_{j j'} (\tilde\beta^*)|  > \frac{n\lambda_n}{3 (m+s) (\log n)^3} \right\} + o(1), \nonumber\\
&\leq & \sum_{j \in (\mMS)^c} \sum_{j' \in \mMS} \E [| \bm H_{j j'} (\tilde\beta^*)|] \frac{3 (m+s) (\log n)^3}{n\lambda_n} + o(1).
\end{eqnarray}
where the second inequality holds by H\"older's inequality and the fourth inequality holds by Markov's inequality. Further, we show that $\E [| \bm H_{j j'} (\tilde\beta^*)|]$ is upper bounded. By elementary calculus, we obtain that
\begin{eqnarray*}
    \lefteqn{\bm H_{jj'}(\beta) = \frac{1}{n}\sum_{i =1}^n \sum_{k =1}^l X_{ikj}X_{ik j'}\Big[\left(Y_{ik} - g(X_{ik}^\top\beta)\right) \ddot g(X_{ik}^\top\beta) [\bSigma\left(\bX_{i,\mA}\right)^{-1}]_{kk}}&&\\
    && - (\ddot g(X_{ik}^\top\beta) [\bSigma\left(\bX_{i,\mA}\right)^{-1}]_{kk} + \left(Y_{ik} - g(X_{ik}^\top\beta)\right) \ddot g(X_{ik}^\top \beta) \dot g(X_{ik}^\top\beta )[\bSigma\left(\bX_{i,\mA}\right)^{-1}]_{kk}.
    \Big]
\end{eqnarray*}
By Assumptions~\ref{cond: active set and non active set covariance}, \eqref{cond: conditional mean function}, \ref{cond: variance function bounded}, it's easy to see that $\E[|\bm H_{jj'}(\beta)|]$ is upper bounded and order is determined by $\E|Y_{ik} - g(X_{ik}^\top\tilde\beta^*)| $. It holds that
\begin{eqnarray*}
    \lefteqn{\E|Y_{ik} - g(X_{ik}^\top\tilde\beta^*)| =\E|Y_{ik} - g(X_{ik}^\top\beta_0) +g(X_{ik}^\top\beta_0)- g(X_{ik}^\top\tilde\beta^*)| } &&\\
    &\leq & \E |Y_{ik} - g(X_{ik}^\top\beta_0)| +|g(X_{ik}^\top\beta_0)- g(X_{ik}^\top\tilde\beta^*)| =  \E |R_{ik} (\beta_0)| + |g(X_{ik}^\top\beta_0)- g(X_{ik}^\top\tilde\beta^*)| \\
    &\leq& C + L|X^\top_{ik} (\beta_0 - \tilde\beta^*)| \leq C + L \|X_{ik}\|_2  \|\beta_0 - \tilde\beta^*\|_2 \leq  C + C' \sqrt{\frac{m+s}{n}},
\end{eqnarray*}
where the third inequality holds by Assumptions~\ref{cond: residuals moments and subgaussian}, Lipschitz continuity of $g$ in \ref{cond: conditional mean function}; the last inequality holds by \ref{cond: fixed and bounded design} and $\tilde\beta^*$ defined on $\mathcal N$. Thus, by Assumption~\ref{cond: log p and n order}, we conclude that 
\begin{eqnarray}
    T_5 \leq O (\frac{p(m+s)^{5/2} (\log n)^3}{n^{3/2}\lambda_n}) + o(1) = o(1).
\end{eqnarray}

\subsection{Proof of Theorem~\ref{thm: covariance matrix}}
\begin{proof}
\begin{itemize}
\item[(i)] \emph{Residual consistency.}

Condition on $\{\check\beta_{(\mMS)^c} = 0\}$, by definition, the estimated residual is
\[
R_{ik}(\check\beta) = Y_{ik} - g(X_{ik}^\top \check\beta).
\]

Expanding $g(X_{ik}^\top \check\beta)$ around $\beta_0$ by a first-order Taylor expansion,
\[
g(X_{ik}^\top \check\beta) 
= g(X_{ik}^\top \beta_0) 
+ \dot g(X_{ik}^\top \tilde\beta)\,X_{ik,\mMS}^\top(\check\beta - \beta_0),
\]
for some $\tilde\beta$ between $\check\beta$ and $\beta_0$.  
Hence
\[
R_{ik}(\check\beta) - R_{ik}(\beta_0)
= - \dot g(X_{ik}^\top \tilde\beta)\,X_{ik}^\top(\check\beta - \beta_0).
\]
By Assumption~\ref{cond: fixed and bounded design}, $\|X_{ik,\mMS}\|_2 = O(\sqrt{s+m})$ uniformly for all $i$ and $k$.  
By Assumption~\ref{cond: conditional mean function}, $\dot g(\cdot)$ is bounded.  
Since $\|\check\beta - \beta_0\|_2 = O_P(\sqrt{(s+m)/n})$, it follows that
\[
\max_{i,k} |R_{ik}(\check\beta) - R_{ik}(\beta_0)| = O_P\left(\frac{s+m}{\sqrt{n}}\right) = o_P(1).
\]

For fixed $k \in [l]$, we show $\|\hat\theta^v_k - \theta^v_k\|_1 = O_P\left(\frac{s_0(s+m)\log(np)}{\sqrt{n}}\right)$; similar result hold for $\max_j \|\hat{\gamma}_j - \gamma_j\|_1 = O_P\left(s_0\sqrt{\frac{\log p}{n}}\right)$, with a sparsity assumption $\max\left(\|\gamma_j\|_0, \|\theta^v_k\|_0\right) = s_0$.

To obtain $\hat\theta^v_k$, we consider the optimization problem defined in \eqref{eq: theta empirical minimization}.

Note that
\begin{eqnarray*}
    \lefteqn{\frac{1}{n}\sum\limits_{i=1}^n \left(f_v(R_{ik}(\check\beta)) - X_{ik}^{\top}\theta\right)^2  = \frac{1}{n}\sum\limits_{i=1}^n \left(f_v(R_{ik}(\check\beta)) - f_v(R_{ik}(\beta_0)) + f_v(R_{ik}(\beta_0)) - X_{ik}^{\top}\theta\right)^2} &&\\
    &=& \frac{1}{n}\sum\limits_{i=1}^n \left(f_v(R_{ik}(\beta_0))- X_{ik}^{\top}\theta\right)^2 + \frac{1}{n}\sum\limits_{i=1}^n \left(f_v(R_{ik}(\beta_0)) - f_v(R_{ik}(\check\beta))\right)^2 \\
    && - \frac{2}{n}\sum\limits_{i=1}^n \left(f_v(R_{ik}(\beta_0)) - X_{ik}^{\top}\theta\right)\left(f_v(R_{ik}(\check\beta)) - f_v(R_{ik}(\beta_0))\right)\\
    &=& I_1 + I_2 + I_3
\end{eqnarray*}

For $I_2$, note that
\begin{align*}
  \frac{1}{n}\sum\limits_{i=1}^n \left(f_v(R_{ik}(\beta_0)) - f_v(R_{ik}(\check\beta))\right)^2  &\leq \frac{C}{n}\sum\limits_{i=1}^n \left(R_{ik}(\beta_0) - R_{ik}(\check\beta)\right)^2 \ \text{lipschitz condition}\\
  & = \frac{C}{n}\sum\limits_{i=1}^n \left(g(X_{ik}^{\top}\beta) - g(X_{ik}^{\top}\hat{\beta})\right)^2\\
  &\leq \frac{C}{n}\sum\limits_{i=1}^n \left(X_{ik}^{\top}(\beta - \hat{\beta})\right)^2\\
  & = O_P\left((s+m)^2\|\beta-\hat{\beta}\|_2^2\right) = O_P((s+m)^2\delta_n^2)
\end{align*}

By definition, $\hat\theta^v_k$ minimizes the objective function 
$\ell_n(\theta) + \lambda\|\theta\|_1$, where $\ell_n(\theta) = \frac{1}{n}\sum_{i=1}^n
    \left(
    f_v(R_{ik}(\check\beta)^2) - X_{ik}^{\top}\theta
    \right)^2$.
Thus,
\begin{align*}
\ell_n(\hat\theta^v_k) + \lambda \|\hat\theta^v_k\|_1 &\le \ell_n(\theta^v_k) + \lambda \|\theta^v_k\|_1 \\
I_1(\hat\theta^v_k) + I_2 + I_3(\hat\theta^v_k) + \lambda \|\hat\theta^v_k\|_1 &\le I_1(\theta^v_k) + I_2 + I_3(\theta^v_k) + \lambda \|\theta^v_k\|_1
\end{align*}

Let $\Delta = \hat\theta^v_k - \theta^v_k$. Rearranging the terms yields:
\begin{equation}
I_1(\hat\theta^v_k) - I_1(\theta^v_k) \le \big( I_3(\theta^v_k) - I_3(\hat\theta^v_k) \big) + \lambda (\|\theta^v_k\|_1 - \|\hat\theta^v_k\|_1)
\label{eq:start}
\end{equation}

For the left-hand side, we have
\begin{align*}
I_1(\hat\theta^v_k) - I_1(\theta^v_k) &= \frac{1}{n} \sum_{i,k} \left[ (f_v(R_{ik}(\beta_0)) - X_{ik}^\top \hat\theta^v_k)^2 - (f_v(R_{ik}(\beta_0)) - X_{ik}^\top \theta^v_k)^2 \right] \\
&= \frac{1}{n} \sum_{i,k} \left[ (\epsilon_{il} - X_{ik}^\top \Delta)^2 - \epsilon_{il}^2 \right] \\
&= \frac{1}{n} \|\mathbf{X}\Delta\|_2^2 - \frac{2}{n} \langle \mathbf{X}^T \epsilon, \Delta \rangle
\end{align*}

For $I_3$, let $E_{R, ik} = f_v(R_{ik}(\check\beta)) - f_v(R_{ik}(\beta_0))$.
\begin{align*}
I_3(\theta^v_k) - I_3(\hat\theta^v_k) &= \frac{2}{n} \sum_{i} \left[ (f_v(R_{ik}(\beta_0)) - X_{ik}^\top \hat\theta^v_k) - (f_v(R_{ik}(\beta_0)) - X_{ik}^\top \theta^v_k) \right] E_{R, ik} \\
&= -\frac{2}{n} \sum_{i} (X_{ik}^\top \Delta) E_{R, ik} = -\frac{2}{n} \langle \mathbf{X}\Delta, E_R \rangle
\end{align*}

Substituting the expanded terms into Equation \eqref{eq:start}:
$$
\frac{1}{n} \|\mathbf{X}\Delta\|_2^2 - \frac{2}{n} \langle \mathbf{X}^T \epsilon, \Delta \rangle \le -\frac{2}{n} \langle \mathbf{X}\Delta, E_R \rangle + \lambda (\|\theta^v_k\|_1 - \|\hat\theta^v_k\|_1)
$$
Rearranging gives:
$$
\frac{1}{n} \|\mathbf{X}\Delta\|_2^2 \le \frac{2}{n} \langle \mathbf{X}^T(\epsilon - E_R), \Delta \rangle + \lambda (\|\theta^v_k\|_1 - \|\hat\theta^v_k\|_1)
$$
We bound the first term on the RHS using the H\"older inequality and our choice of $\lambda$. We choose the tuning parameter $\lambda$ such that $\lambda \ge \frac{4}{n} \|\mathbf{X}^T(\epsilon - E_R)\|_{\infty}$. This yields:
$$
\left|\frac{2}{n} \langle \mathbf{X}^T(\epsilon - E_R), \Delta \rangle\right| \le \frac{2}{n} \|\mathbf{X}^T(\epsilon - E_R)\|_{\infty} \|\Delta\|_1 \le \frac{\lambda}{2}\|\Delta\|_1
$$
Let $S_0$ be the support of $\theta^v_k$. Using the standard property that $\|\theta^v_k\|_1 - \|\hat\theta^v_k\|_1 \le \|\Delta_{S_0}\|_1 - \|\Delta_{S_0^c}\|_1$, we obtain:
\begin{align*}
\frac{1}{n} \|\mathbf{X}\Delta\|_2^2 &\le \frac{\lambda}{2}\|\Delta\|_1 + \lambda(\|\Delta_{S_0}\|_1 - \|\Delta_{S_0^c}\|_1) \\
&= \frac{\lambda}{2}(\|\Delta_{S_0}\|_1 + \|\Delta_{S_0^c}\|_1) + \lambda(\|\Delta_{S_0}\|_1 - \|\Delta_{S_0^c}\|_1) \\
&= \frac{3\lambda}{2}\|\Delta_{S_0}\|_1 - \frac{\lambda}{2}\|\Delta_{S_0^c}\|_1
\end{align*}
This leads to the final basic inequality, which is central to the LASSO theory:
$$
\frac{1}{n} \|\mathbf{X}\Delta\|_2^2 + \frac{\lambda}{2}\|\Delta_{S_0^c}\|_1 \le \frac{3\lambda}{2}\|\Delta_{S_0}\|_1
$$

This inequality implies that $\|\Delta_{S_0^c}\|_1 \le 3\|\Delta_{S_0}\|_1$, which is the cone condition required for the RE property to hold. Assuming that for some $\kappa > 0$, we have $\frac{1}{n}\|\mathbf{X}\Delta\|_2^2 \ge \kappa \|\Delta\|_2^2$ for all $\Delta$ in this cone, then:
$$
\kappa \|\Delta\|_2^2 \le \frac{1}{n} \|\mathbf{X}\Delta\|_2^2 \le \frac{3\lambda}{2}\|\Delta_{S_0}\|_1
$$
Using Cauchy-Schwarz, $\|\Delta_{S_0}\|_1 \le \sqrt{s_0}\|\Delta_{S_0}\|_2 \le \sqrt{s_0}\|\Delta\|_2$, where $s_0 = |S_0|$ is the sparsity.
$$
\kappa \|\Delta\|_2^2 \le \frac{3\lambda}{2}\sqrt{s_0}\|\Delta\|_2
$$
Dividing by $\|\Delta\|_2$ gives the final rate for the $\ell_2$-norm of the error:
$$
\|\hat\theta^v_k - \theta^v_k\|_2 = \|\Delta\|_2 \le \frac{3\lambda\sqrt{s_0}}{2\kappa}
$$
The rate is determined by the magnitude of $\lambda$. Our choice of $\lambda$ depends on both the stochastic noise $\epsilon$ and the measurement error $E_R$. Typically, $\|\frac{1}{n}\mathbf{X}^T\epsilon\|_\infty = O_P(\sqrt{\log p/n})$. The final task is to bound $\|\frac{1}{n}\bX^{\top}E_R\|_{\infty}$. Note that
\begin{align*}
\|\frac{1}{n}\bX^{\top}E_R\|_{\infty} &= \max_j \frac{1}{n}|\sum\limits_{i=1}^n\sum_{k=1}^l X_{ik,j}(f_v(R_{ik}(\check\beta)) - f_v(R_{ik}(\beta_0)))|\\
&\leq \frac{C}{n}\max_{ik,j}|X_{ik,j}|\sum_{i=1}^n\sum_{k=1}^l |X_{ik}^{\top}(\hat{\beta}-\beta)|\\
&\leq \frac{C}{n}\sqrt{\log (np)}\sum_{i=1}^n\sum_{k=1}^l \|X_{ik}\|_{2}\|\hat{\beta}-\beta\|_2\\
&\leq \sqrt{s+m}\log (np) \|\hat{\beta}-\beta\|_2
\end{align*}

Therefore, we must set $\lambda \asymp \sqrt{\frac{\log p}{n}} + \sqrt{s+m}\log(np)\delta_n$. Since $\delta_n = \sqrt{\frac{s+m}{n}}$, the final convergence rate is:
\begin{eqnarray}\label{eq: theta hat rate}
   \|\hat\theta^v_k - \theta^v_k\|_2 = O_P\left((s+m)\log(np)\sqrt{\frac{s_0}{n}}\right) . 
\end{eqnarray}
This explicitly shows that the final error is composed of the standard high-dimensional estimation error and an additional term propagated from the error in estimating $R_{ik}(\beta_0)$.

For $\hat{\gamma}_j$, a simple modification of Theorem 5.2 in van de Geer and Muller (2012) implies that
\begin{eqnarray}\label{eq: gammahat rate}
    \left\|\hat{\gamma}_j-\gamma\right\|_b & = O_P \left(s^{1 / b} \sqrt{\frac{\log p}{n}}\right), b=1,2.
\end{eqnarray}
Thus, it follows that
\begin{eqnarray}\label{eq: gammahat prediction rate}
   \frac{1}{n} \sum_{i=1}^n \sum_{k=1}^l\left[X_{i k}^{\top}\left(\hat{\gamma}-\gamma\right)\right]^2 & =O_P\left(\frac{s_0 \log p}{n}\right) .
\end{eqnarray}

\item[(ii)] \emph{Consistency of active set selection.} 

We define the following sets $\mbox{FP}(t) = \sum_{k=1}^l\sum_{j \in \mA^c} I( W_{kj} > t)$.
We first show that $P\left(\sum_{j \in \mA^c} I( W_{kj} > t_0) = 0\right) \to 0$ with the critical value $t_0 = \chi_h^2 (1- \alpha_p / p)$ for $\alpha_p = p^{-c}$ for some $c >0$. 
\begin{eqnarray}\label{eq: type-I error single k calculation}
    P\left(\sum_{j \in \mA^c}I( W_{kj} > t_0 ) > 0 \right) &=&
    P \left(\cup_{j \in \mA_k^c} W_{kj} > \chi_h^2 (1 - \alpha/p)  \right) \nonumber\\
    &\leq& (p-s_\mA) \frac{\alpha_p}{p} = O(p^{-c}) \to 0,
\end{eqnarray}
where the inequality holds by union bound and the second equality holds by the definition of $\alpha_p$. 
Further, it holds that
\begin{eqnarray}\label{eq: power approximation}
    \lefteqn{P\left(\cap_{j \in \mA_k} (W_{kj} > t_0) \right) = P\left(\min_{j\in\mA_k} W_{kj}  > t_0 \right) }&& \nonumber\\
    &=& \Phi\left(- \frac{t_0 - (h+ \min_{j \in \mA_k}\delta_{kj} B^\top_{kj} \Omega_{kj} B_{kj})}{\sqrt{2(h+ 2 \min_{j \in \mA_k}\delta_{kj} B^\top_{kj} \Omega_{kj} B_{kj} )}} \right) + o(1),
\end{eqnarray}
where the second equality holds by Gaussian approximation to $\chi^2_h (\delta_{kj} B^\top_{kj} \Omega_{kj} B_{kj})$. Moreover, by the boundedness of the eigenvalues of $\Omega_{kj}$ implied by
Assumption~\ref{cond: sufficient dimension reduction}, the noncentrality parameter satisfies
\begin{eqnarray*}
    \delta_{kj} B_{kj}^\top \Omega_{kj} B_{kj}
    &\ge&
    \delta_{kj} \lambda_{\min}(\Omega_{kj}) \|B_{kj}\|_2^2 \\
    &=&
    \delta_{kj} \lambda_{\max}(\Omega_{kj}^{-1})^{-1} \|B_{kj}\|_2^2 .
\end{eqnarray*}

If the minimal signal strength obeys \[\min_{j \in \mA_k} \|B_{kj}\|_2 \gg  \sqrt{2(1 + c) \log p \max_j \frac{\lambda_{\max}(\Omega_{kj}^{-1})}{\delta_{kj}}},\] then, by Assumption~\ref{cond: sufficient dimension reduction}, it follows that
$\min_{j\in \mA_k}\delta_{kj} B^\top_{kj} \Omega_{kj} B_{kj} \gg 2(1 +c) \log p$. Since $t_0 = \chi_h^2(1 - \alpha_p/p) = h+ 2 \log (p/\alpha_p) + o(1)$, the asymptotic power in \eqref{eq: power approximation} converges to 1 for $\min_{j \in \mA_k}\delta_{kj} B^\top_{kj} \Omega_{kj} B_{kj} \gg 2\log (p/\alpha_p) = 2(1 + c) \log p $. 
The estimated active set $\hat \mA = \cup_{k \in [l]} \hat\mA_k$, since $l$ is fixed and by \eqref{eq: type-I error single k calculation}, the size of the test
\begin{eqnarray}
    P ( \mbox{FP}(t_0) >0 ) \leq l \cdot P\left(\sum_{j \in \mA^c}I( W_{kj} > t_0 ) > 0 \right) \to 0.
\end{eqnarray}
Further, the asymptotic power can be expressed as
\begin{eqnarray*}
  \lefteqn{P( \mA \subseteq \hat \mA) = P\left(\cup_{k \in [l]}(\cap_{j \in \mA_k} \left(W_{kj} > t_0) \right)\right)}&&\\ 
  &=& P\left(\cup_{k \in [l]} (\min_{j\in \mA_k} W_{kj} > t_0)\right) 
    = P\left( \min_{j\in \mA_k k\in [l]} W_{kj} > t_0)\right).
\end{eqnarray*}
With similar argument as \eqref{eq: power approximation}, $P( \mA \subseteq \hat \mA) \to 1$ if $\min_{j \in \mA_k, k \in [l]}\delta_{kj} B^\top_{kj} \Omega_{kj} B_{kj} \gg 2\log (p/\alpha_p) = 2(1 + c) \log p $. 
\end{itemize}
\end{proof}

\subsection{Proof of Proposition~\ref{thm: oracle property}}\label{ssec: proof of consistency of covariance matrix}
\begin{proof}
The consistency results in (i) and (ii) hold by Proposition~\ref{prop: consistency} and we focus on showing the asymptotic normality in \eqref{eq: asymptotic normal of subvector}. 

By Taylor expansion, we obtain 
\begin{eqnarray}\label{eq: taylor for asymptotic normal}
   \lefteqn{\bm 0 = U_{\mMS} (\tilde\beta) = U_{\mMS}(\beta_0) + \frac{\partial U_{\mMS}(\beta_0)}{\partial \beta_{\mMS}} (\tilde\beta_{\mMS} - \beta_{0, \mMS})} &\\
   & + \frac{1}{2}(\tilde\beta_{\mMS}- \beta_{0, \mMS})^\top \frac{\partial^2 U_{\mMS}(\beta^*)}{\partial \beta_{\mMS} \partial \beta_{\mMS}^\top}(\tilde\beta_{\mMS} - \beta_{0, \mMS}), \nonumber
\end{eqnarray}
where $\beta^*$ is between $\tilde\beta$ and $\beta_0$ and the third term on the RHS is a remainder term. 
We first show that the remainder term is negligible when $s + m = o(n^{1/3})$. Notice that the Hessian $\frac{\partial^2 U_{\mMS}(\beta^*)}{\partial \beta_{\mMS} \partial \beta_{\mMS}^\top} \in \mathbbm R^{(s + m) \times (s+m) \times (s+m)}$ has the operator norm
\begin{eqnarray*}
    \left\|\frac{\partial^2 U_{\mMS}(\beta^*)}{\partial \beta_{\mMS} \partial \beta_{\mMS}^\top} \right\|_{\rm op} \coloneqq \sup_{\|u \|_2 = \|v\|_2 = 1} \left\|u^\top \frac{\partial^2 U_{\mMS}(\beta^*)}{ \partial \beta_{\mMS} \partial \beta_{\mMS}^\top} v \right\|_2 ,
\end{eqnarray*}
where
\begin{eqnarray*}
\lefteqn{\frac{\partial^2 U_{\mMS}(\beta^*)}{\partial \beta_{\mMS} \partial \beta_{\mMS}^\top}} &&\\ 
&=& -\frac{1}{n}\sum_{i=1}^n \bX_{i, \mMS} 
\Big[\ddot \bD_i(\beta^*)\bSigma\left(\bX_{i,\mA}\right)^{-1} \bD_i(\beta^*) 
+ \bD_i(\beta^*)\bSigma\left(\bX_{i,\mA}\right)^{-1} \ddot \bD_i(\beta^*) 
\\
&&+ \ddot \bD_i(\beta^*)\bSigma\left(\bX_{i,\mA}\right)^{-1} \bD_i(\beta^*)\Big] \bX_{i, \mMS}^\top \\
&&+ \frac{1}{n}\sum_{i=1}^n \bX_{i, \mMS} \dddot \bD_i(\beta^*)\bSigma\left(\bX_{i,\mA}\right)^{-1}\{Y_i-\bg(\bX_i^\top \beta^*)\} \bX_{i, \mMS}^\top,\\
&\coloneqq & \bH_{\mMS, \rm{main}} + \bH_{\mMS, \rm{rem}},
\end{eqnarray*} 
$
\ddot \bD_i(\beta) = \mathrm{diag} \big(\ddot g(\bX_i^\top \beta), \dots, \ddot g(\bX_i^\top \beta)\big)$, and $\dddot \bD_i(\beta) = \mathrm{diag}\big(\dddot g(\bX_i^\top \beta), \dots, \dddot g(\bX_i^\top \beta)\big)$.

Further, the higher order term $\bH_{\mMS, \rm{rem}}$ is of order $O_P \left(\sqrt{\frac{m+s}{n}}\right)$ since the residual term 
\[
\| Y_i - \bg (\bX_i^\top\beta^*) \|_2 \leq  \||Y_i - \bg( \bX_i^\top\beta_0) \|_2 + \|\bg(\bX_i^\top\beta_0) - \bg(\bX_i^\top\beta^*) \|_2;
\]
the first term is of order $O_P(n^{-1/2})$ by Assumptions~\ref{cond: conditional mean function} and the second term is of order $O_P \left(\sqrt{\frac{m+s}{n}}\right)$ by the Lipschitz continuity of $g$ in Assumption~\ref{cond: conditional mean function} and $\beta^*$ is between $\beta_0$ and the oracle solution $\tilde\beta \in \mathcal{N}_{\beta, \tau}$. By Assumptions~\ref{cond: fixed and bounded design}, \ref{cond: conditional mean function}, \ref{cond: variance function bounded}, we conclude that $\bH_{\mMS, \rm{rem}} = O_P \left(\sqrt{\frac{m+s}{n}}\right)$ which is negligible when $s+m = o(n^{1/3})$. We easily argue that the leading term $\bH_{\mMS, \rm{main}} \leq C$ by Assumptions~\ref{cond: fixed and bounded design}, \ref{cond: conditional mean function}, \ref{cond: variance function bounded}. Thus,
\[
\left\|\frac{\partial^2 U_{\mMS}(\beta^*)}{\partial \beta_{\mMS} \partial \beta_{\mMS}^\top}\right\|_{\rm op} = O_P(1) .
\]

The third term in \eqref{eq: taylor for asymptotic normal} satisfies 
\begin{eqnarray*}
   \lefteqn{\sqrt{n}\left\|\frac{1}{2}(\tilde\beta_{\mMS}- \beta_{0, \mMS})^\top \frac{\partial^2 U_{\mMS}(\beta^*)}{\partial \beta_{\mMS} \partial \beta_{\mMS}^\top}(\tilde\beta_{\mMS} - \beta_{0, \mMS})\right\|_2} && \\
   &\leq&\frac{1}{2} \sqrt{n} \left\|\tilde\beta_{\mMS}- \beta_{0, \mMS} \right\|_2^2 \left\|\frac{\partial^2 U_{\mMS}(\beta^*)}{\partial \beta_{\mMS} \partial \beta_{\mMS}^\top} \right\|_{\rm op} = O_P(\frac{s+m}{n^{1/2}}),
   \end{eqnarray*}
which is negligible when $s + m = o(n^{1/2})$.

By rearranging \eqref{eq: taylor for asymptotic normal}, it hold that 
\begin{eqnarray}\label{eq: linearization}
    \lefteqn{\sqrt{n} (\tilde\beta_{\mMS} - \beta_{0, \mMS}) = - \sqrt{n} (\frac{\partial U_{\mMS}(\beta_0)}{\partial\beta_{\mMS}})^{-1}(U_{\mMS} (\beta_0) ) + o_P(1) } && \nonumber\\
    &= & (- \frac{\partial U_{\mMS}(\beta_0)}{\partial\beta_{\mMS}})^{-1}(\frac{1}{\sqrt{n}} \sum_{i = 1}^n U_{\mMS, i} (\beta_0) ) + o_P(1) 
\end{eqnarray}
By Assumption~\ref{cond: residuals moments and subgaussian}, it holds that $\E [\frac{1}{\sqrt{n}}U_{\mMS, i} (\beta_0) \mid \bX = \bX_i] = \bm 0$ and 
\begin{eqnarray*}
\lefteqn{
\Var ( \frac{1}{\sqrt{n}}U_{\mMS, i} (\beta_0) \mid \bX = \bX_i)(\beta_0)}&&\\
&=& \frac{1}{n}\bX_{i,\mMS} \bD_i(\beta_0)\bSigma\left(\bX_{i,\mA}\right)^{-1}
\Var(Y_i \mid \bX_i)\,
\bSigma\left(\bX_{i,\mA}\right)^{-1}\bD_i(\beta_0)\bX_{i,\mMS}^\top. 
\end{eqnarray*}
By the Bentkus Theorem \citep{bentkus2005lyapunov}, 
\begin{equation}\label{eq: asymptotic normality bentkus}
    \frac{1}{\sqrt{n}} \sum_{i = 1}^n U_{\mMS, i} (\beta_0) \stackrel{d}{\to}
\mathcal{N}_{s + m } \left(\bm 0, \mathbf{V}_{2} (\beta_0) \right),
\end{equation}
where the variance is a deterministic term by Assumptions~\ref{cond: fixed and bounded design} and \ref{cond: residuals moments and subgaussian} and is as follows
\begin{eqnarray*}
\lefteqn{\mathbf{V}_{2}(\beta_0) = \lim_{n\to \infty}\frac{1}{n} \sum_{i=1}^n }&\\
&\bX_{i, \mMS} \bD_i(\beta_0) \bSigma\left(\bX_{i,\mA}\right)^{-1} \E \left[\left\{Y_i - \bg(\bX_i^\top \beta_0) \right\}\left\{Y_i - \bg(\bX_i^\top \beta_0) \right\}^{\top}\right] \bSigma\left(\bX_{i,\mA}\right)^{-1} \bD_i(\beta_0)\bX_{i, \mMS}^{\top}.   
\end{eqnarray*}

Further, 
\begin{eqnarray}\label{eq: sandwich bread consitent}
\lefteqn{ -\frac{\partial U_{\mMS}(\beta_0)}{\partial\beta_{\mMS}} \stackrel{p}{\to}
{\mathbf{V}}_{1}(\beta_0)}&&\\
&&=\lim_{n\to\infty} \frac{1}{n}\sum\limits_{i=1}^n \bX_{i,\mMS} \bD_i(\beta_0) \bSigma\left(\bX_{i,\mA}\right)^{-1} \bD_i \left(\beta_0 \right)\bX_{i, \mMS}^{\top}.    \nonumber
\end{eqnarray}

Combine \eqref{eq: linearization}, \eqref{eq: asymptotic normality bentkus}, \eqref{eq: sandwich bread consitent}, the asymptotic normality in \eqref{eq: asymptotic normal of subvector} holds.

\end{proof}

\subsection{Proof of Lemma~\ref{thm: estimator property}}\label{ssec: proof of Theorem 2}
\begin{proof}
Consider the decomposition
\begin{align}
    \hat{U}_n(\beta) + \partial\rho_\lambda(\beta;\mM) = \Delta_n (\beta) + U_n(\beta) + \partial \rho_\lambda(\beta;\mM),
\end{align}
where the approximation error term is given by
\begin{equation}
    \Delta_n (\beta) = \hat{U}_n(\beta) - U_n(\beta) = \frac{1}{n}\sum_{i=1}^n \bX_{i,\mMS}\bD_i(\beta)\left(\hat{\bSigma}_i^{-1} - \bSigma_i^{-1}\right)\{Y_i - \bg(\bX_i^{\top}\beta)\}.
\end{equation}
For notational brevity, we suppress the dependence on $(\bX_{i, \mA}, \beta_0)$ and denote $\hat{\bSigma}(\bX_{i, \mA}; \beta_0)$ simply as $\hat{\bSigma}_i$.

Proposition~\ref{thm: oracle property} established the existence of the oracle estimator as a root of $\bm 0 \in U_n(\beta) + \partial\rho_\lambda(\beta;\mM) $. To prove the existence of a solution to the perturbed equation $\bm 0 \in \hat{U}_n(\beta) + \partial\rho_\lambda(\beta;\mM) $, our strategy is to demonstrate that the term $\Delta_n(\beta)$ is asymptotically negligible relative to the leading terms.

Specifically, we require the following properties to hold. These terms play roles analogous to $T_1$--$T_5$ in the proof of Proposition~\ref{prop: consistency}, and we denote these terms to be 
\begin{align}
J_1 
&\coloneqq \|\Delta_{n,\mMS}(\beta_0)\|_2
= o_P\!\left(\sqrt{\frac{s+m}{n}}\right),
\label{eq: prop 1}\\[2mm]
J_{21} 
&\coloneqq (\tilde\beta_{\mMS}-\beta_{0,\mMS})^\top
\big\{\nabla_{\mMS}\Delta_{n,\mMS}(\beta_0)\big\}
(\tilde\beta_{\mMS}-\beta_{0,\mMS})
= o_P\!\left(\frac{s+m}{n}\right),
\label{eq: prop 2}\\[2mm]
J_{22}
&\coloneqq (\tilde\beta_{\mMS}-\beta_{0,\mMS})^\top
\Big\{\nabla_{\mMS}\Delta_{n,\mMS}(\beta^*)
-\nabla_{\mMS}\Delta_{n,\mMS}(\beta_0)\Big\}
(\tilde\beta_{\mMS}-\beta_{0,\mMS}) \notag\\
&= o_P\!\left(\frac{s+m}{n}\right),
\label{eq: prop 3}\\[2mm]
J_3 
&\coloneqq \max_{j\in(\mMS)^c}\big|\Delta_{n,j}(\beta_0)\big|
= o_P\!\left(\frac{\lambda_n}{\log n}\right),
\label{eq: prop 4}\\[2mm]
J_4 
&\coloneqq \max_{j\in(\mMS)^c}\left|
\nabla_{\mMS}\Delta_{n,j}(\beta_0)\,(\tilde\beta_{\mMS}-\beta_{0,\mMS})
\right|
= o_P\!\left(\frac{\lambda_n}{\log n}\right),
\label{eq: prop 5}\\[2mm]
J_5 
&\coloneqq \max_{j\in(\mMS)^c}\left|
(\tilde\beta_{\mMS}-\beta_{0,\mMS})^\top
\Delta\bH_{j,\mMS}(\beta)\,
(\tilde\beta_{\mMS}-\beta_{0,\mMS})
\right|
= o_P\!\left(\frac{\lambda_n}{\log n}\right),
\label{eq: prop 6}
\end{align}
where $\Delta \bH_{j, \mMS} (\beta) \coloneqq \frac{\partial^2 \Delta_{n,j}(\beta)}{\partial \beta_{\mMS} \partial \beta_{\mMS}^\top}$.

We first show \eqref{eq: prop 1}. By conditioning on the auxiliary subsample, which ensures that $\hat{\bSigma}_i$ is independent of $Y_i$, we obtain $\mathbb{E}[\Delta_n(\beta_0)] = \bm{0}$. For notational simplicity, let us denote $\mathbf{A}_i = \hat{\bSigma}_i^{-1} - \bSigma_i^{-1}$. Then, we have
\begin{eqnarray}\label{eq: estimating equation Delta_n^2}
\lefteqn{\E\left\{\left\|\Delta_{n,\mMS}\left(\beta_{0}\right)\right\|^{2}_2\right\} 
} && \\
& = & \E\left\{\left\|\frac{1}{n} \sum\limits_{i=1}^n\bX_{i, \mMS} \bD\left(\beta_0\right)\bA_i\left\{Y_i - \bg\left(\bX_i^{\top}\beta_0\right)\right\}\right\|_2^2\right\}  \nonumber\\
& = & \E\left\{\frac{1}{n^2} \sum\limits_{i=1}^n\left\|\bX_{i, \mMS} \bD\left(\beta_0\right)\bA_i\left\{Y_i - \bg\left(\bX_i^{\top}\beta_0\right)\right\}\right\|_2^2\right\} \nonumber\\
& = & \mbox{tr}\bigg\{\E\Big\{\frac{1}{n^2} \sum\limits_{i=1}^n\bX_{i, \mMS} \bD\left(\beta_0\right)\bA_i\left\{Y_i - \bg\left(\bX_i^{\top}\beta_0\right)\right\}  \nonumber\\
&& \qquad\qquad \left\{Y_i - \bg\left(\bX_i^{\top}\beta_0\right)\right\}^{\top}\bA_i\bD\left(\beta_0\right)\bX_{i, \mMS}^{\top}\Big\} \bigg\}   \nonumber\\
& = &  \frac{1}{n^2} \mbox{tr}\left\{\sum\limits_{i=1}^n\bX_{i, \mMS} \bD\left(\beta_0\right)\bA_i\bSigma\left(\bX_{i,\mA}\right)\bA_i\bD\left(\beta_0\right)\bX_{i, \mMS}^{\top} \right\}\nonumber\\
&\leq & \frac{C\zeta_n^2}{n^2} \mbox{tr}\{\sum\limits_{i=1}^n \bX_{i, \mMS} \bX^\top_{i, \mMS} \} \nonumber\\
&= &  \frac{C\zeta_n^2}{n^2} \sum\limits_{i=1}^n  X^\top_{ik,\mMS}X_{ik, \mMS}  \nonumber\\
&\leq & \frac{C\zeta_n^2}{n^2} \sum\limits_{i=1}^n\sum\limits_{k=1}^ l(s+m) \nonumber\\
&=& o_P(\frac{m+s}{n})   \nonumber,
\end{eqnarray}
where the first inequality follows from the uniform bound $\zeta_n \coloneq \sup_x \|\mathbf{A}_i\|_F =O_P\left(\left(\frac{\log n}{n}\right)^{\frac{\nu}{4\nu + 2l \cdot |\mathcal{A}|}}\right)$ in \eqref{eq: kernel estimator rate}, and condition~\ref{cond: conditional mean function}--\ref{cond: variance function bounded}.

We now show \eqref{eq: prop 2}, that is $J_{21}=o_P((s+m)/n)$. By the Cauchy-Schwarz inequality,
\begin{eqnarray}\label{eq: prop 2 step 1}
    |J_{21}|
\le \|\tilde\beta_{\mMS}-\beta_{0,\mMS}\|_2^2\;
\big\|\nabla_{\mMS}\Delta_{n,\mMS}(\beta_0)\big\|_{F}.
\end{eqnarray}
Note that differentiating $\Delta_{n,\mMS}(\beta_0)$ introduces only bounded factors by Assumptions~\ref{cond: fixed and bounded design}, \ref{cond: conditional mean function}, \ref{cond: residuals moments and subgaussian}; thus, 
\[
\big\|\nabla_{\mMS}\Delta_{n,\mMS}(\beta_0)\big\|_{F}
\leq C \zeta_n,
\]
where $C$ is some constant.
Combine \eqref{eq: prop 2 step 1}, \eqref{eq: kernel estimator rate}, and \eqref{eq: prop active set estimation consistency}, we obtain
\[
|J_{21}|
\le  C\zeta_n \cdot O_P\Big(\frac{s+m}{n}\Big)
=o_P\Big(\frac{s+m}{n}\Big).
\]

To prove \eqref{eq: prop 3}, recall that $\beta^*$ is between $\tilde\beta$ and $\beta_0$ and let $\bH_{\mMS}(\beta)$ be the Hessian of $\Delta_{n, \mMS}(\beta)$ with respect to $\beta_{\mMS}$.
\begin{eqnarray*}
    \lefteqn{|J_{22}|
=|(\tilde\beta_{\mMS}-\beta_{0,\mMS})^\top
\{\nabla_{\mMS}\Delta_{n,\mMS}(\beta^*)-\nabla_{\mMS}\Delta_{n,\mMS}(\beta_0)\}
(\tilde\beta_{\mMS}-\beta_{0,\mMS})|} &\\
&\le\big\|\nabla_{\mMS}\Delta_{n,\mMS}(\beta^*)-\nabla_{\mMS}\Delta_{n,\mMS}(\beta_0)\big\|_{\op} \|\tilde\beta_{\mMS}-\beta_{0,\mMS}\|_2^2 \\
&\le
\sup_{\beta\in\mathcal N}\|\bH_{\mMS}(\beta)\|_{\op}\cdot \|\beta^*_{\mMS}-\beta_{0, \mMS}\|_2 \|\tilde\beta_{\mMS}-\beta_{0,\mMS}\|_2^2\\
&\le
\sup_{\beta\in\mathcal N}\|\bH_{\mMS}(\beta)\|_{F}\cdot \|\beta^*_{\mMS}-\beta_{0, \mMS}\|_2 \|\tilde\beta_{\mMS}-\beta_{0,\mMS}\|_2^2,\\
&\le C' \zeta_n \cdot O_P\Big(\sqrt{\frac{s+m}{n}}\Big) O_P\Big(\frac{s+m}{n}\Big)
=
O_P\Big(\zeta_n\Big(\frac{s+m}{n}\Big)^{3/2}\Big)
=
o_P\Big(\frac{s+m}{n}\Big),
\end{eqnarray*}
where the first inequality holds by the Cauchy-Schwarz, the second by the mean-value expansion. The second equality holds by arguing that differentiating $\Delta_n (\beta)$ twice keeps only bounded factors and a single $A_i$, hence $\sup_{\beta\in\mathcal N}\|\bH_{\mMS}(\beta)\|_{F} \leq C'\zeta_n$ for some constant $C'$; in addition, $\beta^* \in \mathcal{N}_{\beta, \tau}$, and \eqref{eq: prop active set estimation consistency}. The last equality follows from \eqref{eq: kernel estimator rate}.

To show \eqref{eq: prop 4}, we notice that $\Delta_{n,j} = \frac{1}{n} \sum_{i = 1}^n X_{ij} \bD(\beta_0) \bA_i \{Y_i - \bg (\bX_i^\top \beta_0 )\}$ and 
\begin{eqnarray*}
\lefteqn{P\left\{\max_{j \in(\mMS)^{c}}\left|\Delta_{n,j}\left(\beta_{0}\right)\right|> \frac{ \zeta_n\lambda_n}{ \log n}\right\}}&& \\
&\leq& \sum_{j \in (\mMS)^c} P\left\{\left|\frac{1}{n} \sum_{i=1}^n X_{ij} \bD_i \bA_i(Y_i - \bm g(\bX_i^\top\beta_0))\right|
>\frac{\zeta_n\lambda_n}{\log n}\right\} \\
& \leq & 2 p \cdot \exp\left(
- C n \frac{\zeta_n^2\lambda_n^2}{(\log n)^2 \|\bA_i\|_F^2} \right) \\
& \leq & 2 p \cdot \exp\left(
- C n \frac{\lambda_n^2}{(\log n)^2}
\right) = 2\exp \left(
\log p - C n \frac{\lambda_n^2}{(\log n)^2}
\right)
\to 0,
\end{eqnarray*}
where the first inequality holds by union bound, the second holds by Lemma~\ref{lemma: bernstein}, the third holds by the uniform bound of $\bA_i$ in \eqref{eq: kernel estimator rate}, and the convergence holds by Assumption~
\ref{cond: log p and n order}, that is $\log p = o\left\{\frac{n\lambda_n^{2}}{(\log n)^{2}}\right\}$.

We next show \eqref{eq: prop 5}. Recall that
\[
\Delta_{n,j}(\beta)
=
\frac{1}{n}\sum_{i=1}^n
X_{ij}\bD_i(\beta)\bA_i
\big(Y_i-\bm g(\bX_i^\top\beta)\big).
\]

Differentiating with respect to $\beta_{\mMS}$ gives
\[
\nabla_{\mMS} \Delta_{n,j}(\beta)
=
\frac{1}{n}\sum_{i=1}^n \sum_{k=1}^l
X_{ikj}\, A_{ik}\,
\dot g(X_{ik}^\top\beta)\,
X_{ik,\mMS}
\Big[
Y_{ik}
-
2\dot g(X_{ik}^\top\beta)
\Big].
\]
Similar to the proofs for \eqref{eq: prop 2}, \eqref{eq: prop 3}, differentiating
$\Delta_{n,j}(\beta)$ introduces additional bounded factors from $\bD_i(\beta)$ and
$R_i(\beta)$). By componentwise union bound and by Lemma~\ref{lemma: bernstein} setting $t=C'\zeta_n \sqrt{\frac{\log(p(s+m))}{n}}$, we obtain
\begin{equation}\label{eq: prop 5 step 1}
\max_{j}\|\nabla_{\mMS}\Delta_{n,j}(\beta_0)\|_\infty=O_P\left(\zeta_n\sqrt{\frac{\log(p(s+m))}{n}}\right)
=O_P\Big(\zeta_n\sqrt{\frac{\log p}{n}}\Big).    
\end{equation}
Further, by Assumption~\ref{cond: log p and n order} $\log p=o\left(\frac{n\lambda_n^2}{(\log n)^2}\right)$, hence $\frac{\sqrt{\log p}}{n}
=o \left(\frac{\lambda_n}{\sqrt n\,\log n}\right)$. Therefore, \begin{equation}\label{eq: prop 5 step 2}
\zeta_n \frac{(s+m)\sqrt{\log p}}{n}
=
\Big(\zeta_n\frac{s+m}{\sqrt n}\Big)\cdot
o\left(\frac{\lambda_n}{\log n}\right)
=o(1) \cdot o(\left(\frac{\lambda_n}{\log n}\right).    
\end{equation}
Hence,
\begin{eqnarray*}
\lefteqn{J_4= \max_{j\in(\mMS)^c}
\left|
\nabla_{\mMS}\Delta_{n,j}(\beta_0)
(\tilde\beta_{\mMS}-\beta_{0,\mMS})
\right|}&\\
&\le
\left(
\max_{j\in(\mMS)^c}
\|
\nabla_{\mMS}\Delta_{n,j}(\beta_0)
\|_2
\right)
\|
\tilde\beta_{\mMS}-\beta_{0,\mMS}
\|_2 \\ 
&\le \sqrt{s+m} \cdot\max_j\|\nabla_{\mMS}\Delta_{n,j}(\beta_0)\|_\infty \|
\tilde\beta_{\mMS}-\beta_{0,\mMS}
\|_2 \\
&=O_P\Big(\zeta_n\sqrt{\frac{(s+m)\log p}{n}}\Big)\cdot O_P\Big(\sqrt{\frac{s+m}{n}}\Big)
=O_P\Big(\zeta_n\frac{(s+m)\sqrt{\log p}}{n}\Big)
 =o_P\left(
\frac{\lambda_n}{\log n}
\right),
\end{eqnarray*}
where the first inequality follows from Cauchy-Schwarz, the first equality holds by \eqref{eq: prop 5 step 1}, and the last by \eqref{eq: prop 5 step 2}.


Similar to the argument for \eqref{eq: prop 5 step 1}, we argue that by applying union bound and Lemma~\ref{lemma: bernstein} over $j$ gives 
\begin{equation}\label{eq: prop 6 step 1}
\max_j\sup_{\beta\in\mathcal N}\|\Delta\bH_{j,\mMS}(\beta)\|_{\op}
=O_P\Big(\zeta_n\sqrt{\frac{\log p}{n}}\Big).    
\end{equation}

Thus, 
\begin{eqnarray*}
   \lefteqn{|J_5|
\le \| (\tilde\beta_{\mMS}-\beta_{0,\mMS})\|_2^2 \|\Delta\bH_{j,\mMS}(\beta)\|_{\op}}& \\
&\le O_P\Big(\frac{s+m}{n}\Big)\cdot O_P\Big(\zeta_n\sqrt{\frac{\log p}{n}}\Big)
=O_P\Big(\zeta_n\frac{(s+m)\sqrt{\log p}}{n^{3/2}}\Big) = o_P(\lambda_n/\log n),
\end{eqnarray*}
where the second inequality follows by \eqref{eq: prop 5 step 1}, and the last equality holds since $\delta_n\frac{(s+m)\sqrt{\log p}}{n^{3/2}}
=o\Big(\frac{\lambda_n}{\log n}\Big)$ by Assumption~\ref{cond: log p and n order}.

\end{proof}

\subsection{Proof of Theorem~\ref{thm:power-improvement}}\label{ssec: proof of Theorem 3}

\begin{proof}
We exploit the efficiency property of the estimator to obtain the variance inequality. Let $\hat{\mathbf{\Omega}}$ and $\check{\mathbf{\Omega}}$ denote the asymptotic covariance matrices of the combined estimator $\hat{\beta}$ and the initial estimator $\check{\beta}$, respectively. It is well known that the subvector of the cross-fitted estimator $\hat \beta_{\mMS}$ can be expressed through the following asymptotic linear representation
\begin{eqnarray}
    \sqrt{n}(\hat \beta_{\mMS} - \beta_{0, \mMS}) = 
     \bW_1^{-1} \bS_1 + o_P(1),
\end{eqnarray}
where $\bW_1 =  \E[\bX_{i,\mMS}\bD_i(\beta_0) \bSigma (\bX_{i, \mA})^{-1} \bD_i(\beta_0) \bX_{i,\mMS}^{\top}]$, $R_i(\beta_0) = Y_i - \bg(\bX_i^{\top}\beta_0)$, and $\bS_1 = {n}^{-1/2}\sum\limits_{i=1}^n \bX_{i,\mMS}\bD_i(\beta_0)\bSigma(\bX_{i,\mA})^{-1} R_i (\beta_0)$.

Similarly, the initial estimator $\check{\beta}_{\mMS}$ satisfies the representation
\begin{equation}
    \sqrt{n}(\check{\beta}_{\mMS} - \beta_{0, \mMS}) = \bW_2^{-1}
    \bS_2 + o_P(1),
\end{equation}    
where $\bW_2 = \E[ \bX_{i,\mMS}\bD_i(\beta_0) \check\bSigma^{-1} \bD_i(\beta_0) \bX_{i,\mMS}^{\top}]$, and \[\bS_2 = {n}^{-1/2}\sum\limits_{i=1}^n \bX_{i,\mMS}\bD_i (\beta_0) \check\bSigma^{-1}(Y_i - \bg(\bX_i^{\top}\beta_0)).\]

 
The asymptotic cross covariance between $\check{\beta}_{\mMS}$ and $\hat \beta_{\mMS}$ satisfies
    \begin{equation}\label{eq: omega12 1}
        \bOmega_{12} = \E[\left( \bW_2^{-1} \bS_2 (\bW_1^{-1} \bS_1)^\top \right)] = \bW_2^{-1} \mathbb{E}\left[ \bS_2 \bS_1^{\top} \right] \bW_1^{-1}.
    \end{equation}
By independence of observations between subjects $i \neq j$, the cross-terms vanish. We only need to consider the sum of expectations for the $i$-th terms, that is,
\begin{eqnarray}\label{eq: omega12 2}
 \lefteqn{\E\left[ \bS_2 \bS_1^{\top} \right] } &&\\ 
 &=& \frac{1}{n} \sum_{i=1}^n \E \Bigg[ \left(\bX_{i,\mMS} \bD_i(\beta_0) \check\bSigma^{-1} R_i (\beta_0)\right) \left(\bX_{i,\mMS}\bD_i(\beta_0)\bSigma(\bX_{i,\mA})^{-1}R_i (\beta_0)\right)^{\top} \Bigg]  \nonumber \\
 &=& \frac{1}{n} \sum\limits_{i=1}^n \E\left[\bX_{i,\mMS} \bD_i(\beta_0) \check\bSigma^{-1} \E\left[ R_i (\beta_0) R_i (\beta_0)^{\top} \mid \bX_i \right] \left(\bSigma(\bX_{i,\mA})^{-1}\right)^{\top} \bD_i(\beta_0)\bX_{i,\mMS}^{\top} \right], \nonumber\\
 &=& \frac{1}{n}\sum_{i=1}^n \E \left[
\bX_{i,\mMS} \bD_i(\beta_0) \check\bSigma^{-1}
\bSigma(\bX_{i,\mA}) \bSigma(\bX_{i,\mA})^{-1} \bD_i(\beta_0)\bX_{i,\mMS}^{\top} \right]= \bW_2 . \nonumber
\end{eqnarray}
%
Combine \eqref{eq: omega12 1},\eqref{eq: omega12 2}, we obtain 
\begin{eqnarray}\label{eq: cross cov}
    \bOmega_{12} = \bW_2^{-1} \bW_2 \bW_1^{-1} = 
\bW_1^{-1}.
\end{eqnarray}
Similarly, the asymptotic variance of $\hat\beta$ is 
\begin{eqnarray}\label{eq: asymptotic variance oracle} 
   \hat\bOmega = \bW_1 \E[\bS_1 \bS_1^\top ] = \bW_1^{-1} 
\end{eqnarray}
Let $\bm{\delta}=\check\beta_{\mMS}-\hat\beta_{\mMS}$. 
By \eqref{eq: cross cov}, \eqref{eq: asymptotic variance oracle}, we obtain
\[
\mbox{AVar}(\bm{\delta})
=\tilde{\bOmega}-\hat{\bOmega}.
\]
Since $\mbox{AVar}(\bm{\delta})$ is positive semidefinite, it follows that 
$\tilde{\bOmega} \succeq \hat{\bOmega}$. 
For any full row rank matrix $\bC$, premultiplication and postmultiplication preserve the Loewner order,
\[
\bC\tilde{\bOmega}\bC^\top
\succeq
\bC\hat{\bOmega}\bC^\top,
\]
and, when positive definite, inversion reverses the order. Hence under local alternatives with drift $\bh$,
\[
\bh^\top(\bC\hat{\bOmega}\bC^\top)^{-1}\bh
\;\ge\;
\bh^\top(\bC\tilde{\bOmega}\bC^\top)^{-1}\bh,
\]
so the efficient estimator yields weakly larger noncentrality parameter and therefore weakly greater local asymptotic power.

\end{proof}
\section*{Acknowledgements}
	



\bibliographystyle{apalike}

\bibliography{quasilikelihood}   

\begin{thebibliography}{}

\bibitem[Aitkin, 1935]{aitkin1935least}
Aitkin, A.~C. (1935).
\newblock On least squares and linear combination of observations.
\newblock {\em Proceedings of the Royal Society of Edinburgh}, 55:42--48.

\bibitem[Amemiya, 1973]{amemiya1973regression}
Amemiya, T. (1973).
\newblock Regression analysis when the variance of the dependent variable is
  proportional to the square of its expectation.
\newblock {\em Journal of the American Statistical Association},
  68(344):928--934.

\bibitem[Andrews, 1986]{andrews1986note}
Andrews, D. W.~K. (1986).
\newblock A note on the unbiasedness of feasible {GLS}, quasi-maximum
  likelihood, robust, adaptive, and spectral estimators of the linear model.
\newblock {\em Econometrica}, 54(3):687--698.

\bibitem[Bentkus, 2005]{bentkus2005lyapunov}
Bentkus, V. (2005).
\newblock A {Lyapunov}-type bound in {$\mathbb{R}^d$}.
\newblock {\em Theory of Probability \& Its Applications}, 49(2):311--323.

\bibitem[Chernozhukov et~al., 2018]{chernozhukov2018DML}
Chernozhukov, V., Chetverikov, D., Demirer, M., Duflo, E., Hansen, C., Newey,
  W., and Robins, J. (2018).
\newblock Double/debiased machine learning for treatment and structural
  parameters.
\newblock {\em The Econometrics Journal}, 21(1):C1--C68.

\bibitem[Davidian and Carroll, 1987]{davidian1987variance}
Davidian, M. and Carroll, R.~J. (1987).
\newblock Variance function estimation.
\newblock {\em Journal of the American Statistical Association},
  82(400):1079--1091.

\bibitem[Fan and Li, 2001]{fan2001variable}
Fan, J. and Li, R. (2001).
\newblock Variable selection via nonconcave penalized likelihood and its oracle
  properties.
\newblock {\em Journal of the American Statistical Association},
  96(456):1348--1360.

\bibitem[Fang et~al., 2020]{fang2020test}
Fang, E.~X., Ning, Y., and Li, R. (2020).
\newblock Test of significance for high-dimensional longitudinal data.
\newblock {\em Annals of Statistics}, 48(5):2622--2645.

\bibitem[Gin{\'e} and Guillou, 2002]{gine2002rates}
Gin{\'e}, E. and Guillou, A. (2002).
\newblock Rates of strong uniform consistency for multivariate kernel density
  estimators.
\newblock {\em Annales de l'Institut Henri Poincar\'e (B) Probability and
  Statistics}, 38(6):907--921.

\bibitem[Godambe, 1960]{godambe1960optimum}
Godambe, V.~P. (1960).
\newblock An optimum property of regular maximum likelihood estimation.
\newblock {\em The Annals of Mathematical Statistics}, 31(4):1208--1211.

\bibitem[Godambe, 1985]{godambe1985foundations}
Godambe, V.~P. (1985).
\newblock The foundations of finite sample estimation in stochastic processes.
\newblock {\em Biometrika}, 72(2):419--428.

\bibitem[Guo et~al., 2025]{guo2025model}
Guo, X., Li, R., Zhang, Z., and Zou, C. (2025).
\newblock Model-free statistical inference on high-dimensional data.
\newblock {\em Journal of the American Statistical Association},
  120(549):186--197.

\bibitem[Guvenen, 2009]{guvenen2009empirical}
Guvenen, F. (2009).
\newblock An empirical investigation of labor income processes.
\newblock {\em Review of Economic Dynamics}, 12(1):58--79.

\bibitem[Li, 2018]{li2018sufficient}
Li, B. (2018).
\newblock {\em Sufficient dimension reduction: Methods and applications with
  {R}}.
\newblock Chapman and Hall/CRC.

\bibitem[Li, 1991]{li1991sliced}
Li, K.-C. (1991).
\newblock Sliced inverse regression for dimension reduction.
\newblock {\em Journal of the American Statistical Association},
  86(414):316--327.

\bibitem[Liang and Zeger, 1986]{liang1986longitudinal}
Liang, K.-Y. and Zeger, S.~L. (1986).
\newblock Longitudinal data analysis using generalized linear models.
\newblock {\em Biometrika}, 73(1):13--22.

\bibitem[MacKinnon and White, 1985]{mackinnon1985some}
MacKinnon, J.~G. and White, H. (1985).
\newblock Some heteroskedasticity-consistent covariance matrix estimators with
  improved finite sample properties.
\newblock {\em Journal of Econometrics}, 29(3):305--325.

\bibitem[Mammen and van~de Geer, 1997]{mammen1997penalized}
Mammen, E. and van~de Geer, S. (1997).
\newblock Penalized quasi-likelihood estimation in partial linear models.
\newblock {\em The Annals of Statistics}, 25(3):1014--1035.

\bibitem[McCullagh and Nelder, 1989]{McCullaghNelder1989}
McCullagh, P. and Nelder, J.~A. (1989).
\newblock {\em Generalized linear models}.
\newblock Chapman and Hall, 2 edition.

\bibitem[Meghir and Pistaferri, 2004]{meghir2004income}
Meghir, C. and Pistaferri, L. (2004).
\newblock Income variance dynamics and heterogeneity.
\newblock {\em Econometrica}, 72(1):1--32.

\bibitem[Nakagawa et~al., 2025]{nakagawa2025location}
Nakagawa, S., Ortega, S., Gazzea, E., Lagisz, M., Lenz, A., Lundgren, E., and
  Mizuno, A. (2025).
\newblock Location--scale models in ecology and evolution: Heteroscedasticity
  in continuous, count and proportion data.
\newblock {\em Methods in Ecology and Evolution}.

\bibitem[Ortega and Rheinboldt, 2000]{ortega2000iterative}
Ortega, J.~M. and Rheinboldt, W.~C. (2000).
\newblock {\em Iterative solution of nonlinear equations in several variables}.
\newblock SIAM.

\bibitem[Qu et~al., 2000]{qu2000improving}
Qu, A., Lindsay, B.~G., and Li, B. (2000).
\newblock Improving generalized estimating equations using quadratic inference
  functions.
\newblock {\em Biometrika}, 87(4):823--836.

\bibitem[Shi et~al., 2019]{shi2019linear}
Shi, C., Song, R., Chen, Z., and Li, R. (2019).
\newblock Linear hypothesis testing for high dimensional generalized linear
  models.
\newblock {\em Annals of Statistics}, 47(5):2671--2703.

\bibitem[Spady and Stouli, 2018]{spady2018simultaneous}
Spady, R. and Stouli, S. (2018).
\newblock Simultaneous mean-variance regression.
\newblock arXiv:1804.01631.

\bibitem[Spellman et~al., 1998]{spellman1998comprehensive}
Spellman, P.~T., Sherlock, G., Zhang, M.~Q., Iyer, V.~R., Anders, K., Eisen,
  M.~B., Brown, P.~O., Botstein, D., and Futcher, B. (1998).
\newblock Comprehensive identification of cell cycle--regulated genes of the
  yeast {Saccharomyces cerevisiae} by microarray hybridization.
\newblock {\em Molecular Biology of the Cell}, 9(12):3273--3297.

\bibitem[Wang et~al., 2012]{wang2012penalized}
Wang, L., Zhou, J., and Qu, A. (2012).
\newblock Penalized generalized estimating equations for high-dimensional
  longitudinal data analysis.
\newblock {\em Biometrics}, 68(2):353--360.

\bibitem[Wedderburn, 1974]{wedderburn1974quasi}
Wedderburn, R. W.~M. (1974).
\newblock Quasi-likelihood functions, generalized linear models, and the
  {Gauss--Newton} method.
\newblock {\em Biometrika}, 61(3):439--447.

\bibitem[White, 1980]{white1980heteroskedasticity}
White, H. (1980).
\newblock A heteroskedasticity-consistent covariance matrix estimator and a
  direct test for heteroskedasticity.
\newblock {\em Econometrica}, 48(4):817--838.

\bibitem[White, 1982]{White1982}
White, H. (1982).
\newblock Maximum likelihood estimation of misspecified models.
\newblock {\em Econometrica}, 50(1):1--25.

\bibitem[Yin et~al., 2010]{yin2010nonparametric}
Yin, J., Geng, Z., Li, R., and Wang, H. (2010).
\newblock Nonparametric covariance model.
\newblock {\em Statistica Sinica}, 20(1):469--479.

\bibitem[Young and Shah, 2024]{young2024sandwich}
Young, E.~H. and Shah, R.~D. (2024).
\newblock Sandwich boosting for accurate estimation in partially linear models
  for grouped data.
\newblock {\em Journal of the Royal Statistical Society Series B: Statistical
  Methodology}, 86(5):1286--1311.

\bibitem[Zhang, 2010]{zhang2010mcp}
Zhang, C.-H. (2010).
\newblock Nearly unbiased variable selection under minimax concave penalty.
\newblock {\em The Annals of Statistics}, 38(2):894--942.

\end{thebibliography}

\end{document}